\documentclass[final,3p,times]{elsarticle}
\usepackage{amsfonts}
\usepackage{amssymb}
\usepackage{amsmath}
\usepackage{graphics}
\usepackage{graphicx}
\usepackage{epsfig}
\usepackage{epstopdf}
\usepackage{amsthm}
\usepackage{textcomp}
\usepackage{color}
\usepackage{float}
\usepackage{multirow}
\usepackage{subfig}
\biboptions{sort&compress}

\usepackage{algorithmic}
\usepackage{algorithm}

\journal{Journal of Parallel and Distributed Computing}

\begin{document}

\begin{frontmatter}

\title{Assessing the risk of advanced persistent threats}

\cortext[cor1]{Corresponding author}
%%\cortext[cor2]{Principal corresponding author}

\author[rvt]{Xiaofan Yang}
\ead{xfyang1964@gmail.com}

\author[rvt]{Tianrui Zhang}
\ead{363726657@qq.com}

\author[rvt,rvt2]{Lu-Xing Yang\corref{cor1}}
\ead{ylx910920@gmail.com}

\author[rvt3]{Luosheng Wen}
\ead{wls@cqu.edu.cn}

\author[rvt4]{Yuan Yan Tang}
\ead{yytang@umac.mo}

\address[rvt]{School of Software Engineering, Chongqing University, Chongqing, 400044, China}

\address[rvt2]{School of Information Technology, Deakin University, Melbourne, 3125, Australia}

\address[rvt3]{School of Mathematics and Statistics, Chongqing University, Chongqing, 400044, China}

\address[rvt4]{Department of Computer and Information Science, The University of Macau, Macau}

\begin{abstract}

As a new type of cyber attacks, advanced persistent threats (APTs) pose a severe threat to modern society. This paper focuses on the assessment of the risk of APTs. Based on a dynamic model characterizing the time evolution of the state of an organization, the organization's risk is defined as its maximum possible expected loss, and the risk assessment problem is modeled as a constrained optimization problem. The influence of different factors on an organization's risk is uncovered through theoretical analysis. Based on extensive experiments, we speculate that the attack strategy obtained by applying the hill-climbing method to the proposed optimization problem, which we call the HC strategy, always leads to the maximum possible expected loss. We then present a set of five heuristic attack strategies and, through comparative experiments, show that the HC strategy causes a higher risk than all these heuristic strategies do, which supports our conjecture. Finally, the impact of two factors on the attacker's HC cost profit is determined through computer simulations. These findings help understand the risk of APTs in a quantitative manner.

\end{abstract}

\begin{keyword}
advanced persistent threat \sep risk assessment \sep expected loss \sep attack strategy \sep constrained optimization


%\MSC 34D05 \sep 34D20 \sep 34D23 \sep 68M99

\end{keyword}

\end{frontmatter}

% \linenumbers


\section{Introduction}

In this day and age, the functioning of most organizations, ranging from large enterprises and financial institutions to government sectors and military branches, depends heavily on cyber networks interconnecting computer systems. However, these organizations are vulnerable to multifarious cyber attacks. Traditional cyber attacks tended to compromise lots of unspecified computer systems, with the goal of picking low hanging fruits. Conventional cyber defense measures including firewall and intrusion detection have turned out to be effective in withstanding these cyber attacks \cite{Kostopoulos2012, Singer2014}.

The cyber security landscape has changed drastically over the past few years. Many high-profile organizations have experienced a new kind of cyber attacks --- \emph{advanced persistent threats} (APTs) \cite{Virvilis2013}. Compared with traditional attacks, APTs exhibit the following distinctive characteristics: (a) The attacker is a well-resourced and well-organized group, with the goal of stealing as many sensitive data as possible from a specific organization. (b) Based on meticulous reconnaissance, a preliminary advanced social engineering attack is launched on a few target users to gain footholds in the organization's network. (c) More and more systems are infected stealthily and slowly to gain access to critical information, and preys are secretly sent to the attacker \cite{Tankard2011, Cole2013, Wrightson2015}. APTs can evade traditional detection, causing tremendous damage to organizations. In practice, the detection of APTs involves complex analysis of activities in the network of the targeted organization, which is far from mature \cite{Friedberg2015, Marchetti2016}.

Taking a risk-based approach to security has long been the recommended way to secure an organization \cite{Landoll2011, Wheeler2011, Hubbard2016}. The critical shift is that in the past it was recommended but today owing to the APT it is required. In fact, it is no exaggeration to say that everything performed in security should be mapped back to risk and justified by risk \cite{Cole2013}. Normally, we are not going to eliminate a risk, because that would be too expensive or even impossible. Instead, we are going to reduce the risk to an acceptable level, which depends on the critical information we are protecting. When it comes to an APT, the risk taken by the targeted organization translates to the organization's expected loss. When it comes to an organization, it is appropriate to take the worst-case perspective of assessing the risk as the maximum possible expected loss of the organization over all possible APT attacks. To our knowledge, there is no literature on the risk assessment of APTs.

To assess the risk taken by an organization under APTs, the time evolution of the organization's state has to be modeled accurately. Due to the propagation nature of APTs, it is appropriate to characterize the evolution process as an epidemic model \cite{Daley2009, Garetto2003, WenS2013, WenS2014a, WenS2014b, WenS2015}. In view of the persistence of APTs and taking the relevant network into account, the evolution process should be modeled as a differential dynamical system with the network topology. The individual-level dynamical modeling approach, which has been applied to a wide range of areas, ranging from epidemic spreading \cite{Ganesh2005, Draief2006, Mieghem2009, Mieghem2011} and malware spreading \cite{Draief2008, XuSH2012a, XuSH2012b, Sahneh2012, XuSH2014, YangLX2015, YangLX2017a, YangLX2017b, YangLX2017c} to rumor spreading\cite{HeZB2017, YangLX2017d}, meets this requirement. Towards this direction, a number of APT attack-defense models have recently been suggested \cite{ XuSH2015a, ZhengR2015, YangLX2017e}.

This paper addresses the risk assessment of APTs. First, a dynamic model characterizing the time evolution of the security state of an organization is established by employing the individual-level dynamic modeling approach. Then an organization's risk is quantified as its maximum possible expected loss. On this basis, the risk assessment problem boils down to a constrained optimization problem, with the expected loss as the objective function. The influence of different factors on an organization's risk is illuminated through theoretical analysis. Extensive experiments exhibit that an organization's expected loss is unimodal with respect to the attack strategy. Hence, we speculate that the APT attack strategy obtained by applying the hill-climbing method to the proposed optimization problem, which we call the HC strategy, always inflicts the maximum possible expected loss. To validate the conjecture, we formulate a set of five heuristic APT attack strategies. A set of comparative experiments clearly show that the HC strategy causes a higher risk than all the five heuristic strategies do. Hence, our conjecture is corroborated. Finally, the impact of two factors, the attack duration and the attack budget per unit time, on the attacker's HC cost profit is determined through computer simulations. The results obtained help us understand the risk of APTs in a quantitative manner.

The subsequent materials are organized in this fashion. Section 2 measures an organization's risk using its maximum expected loss, and models the risk assessment problem as an optimization problem. Section 3 reveals the influence of different factors on an organization's risk. An attack strategy is proposed in Section 4, which is shown through comparison experiments to cause the maximum expected loss. Section 5 examines the impact of two factors on the attacker's HC cost profit. This work is closed by Section 6.

\newtheorem{de}{Definition}
\newtheorem{expe}{Experiment}
\newtheorem{thm}{Theorem}
\newtheorem{lm}{Lemma}

\section{The modeling of the risk assessment problem}

Suppose some attacker, who represents a well-resourced and well-organized group, is going to conduct an APT campaign on an organization. The organization's defender, who represents the security team affiliated with the organization, faces the following urgent and challenging problem:

\emph{The risk assessment (RA) problem:} Estimate the potential loss of the organization.

This section is dedicated to the modeling of the RA problem. Our modeling process consists of six successive steps: (i) characterize the state of the organization, (ii) describe the defense posture, (iii) formulate the attack strategy, (iv) model the state evolution of the organization, (v) measure the risk of the organization, and (vi) model the RA problem.

\subsection{The state of an organization}

Consider an organization with a set of $N$ computer systems labelled $1, 2, \cdots, N$ interconnected by a network. Let $G = (V, E)$ denote the network, where each node stands for a system, i.e., $V = \{1, 2, \cdots, N\}$, and there is an edge between node $i$ and node $j$, i.e., $\{i, j\} \in E$, if and only if system $i$ can communicate directly with system $j$. Let $\mathbf{A}(G) = \left[a_{ij}\right]_{N \times N}$ denote the adjacency matrix for the network, i.e., $a_{ij} = 1$ or 0 according as $\{i, j\} \in E$ or not.

The \emph{security level} of a node is measured by the amount of the sensitive data stored in the associated system. Let $w_i$ denote the security level of node $i$. In this work, we assume $w_i = d_i$ ($1 \leq i \leq N$), where $d_i = \sum_{j=1}^Na_{ij}$ denotes the degree of node $i$ in the network. This is because a node with a higher degree typically has a higher importance.

In what follows, it is assumed that at any time, each and every node in the network is in one of two possible states: \emph{secure}, i.e, under the defender's control, and \emph{compromised}, i.e, under the attacker's control. Let $X_i(t)$ = 0 and 1 denote the event that node $i$ is secure and compromised at time $t$, respectively. The \emph{state} of the organization at time $t$ is characterized by the vector
\begin{equation}
  \mathbf{X}(t) = (X_1(t), X_2(t), \cdots, X_N(t)).
\end{equation}
Let $S_i(t)$ and $C_i(t)$ denote the probability of the event that node $i$ is secure and compromised at time $t$, respectively.
\begin{equation}
  S_i(t) = \Pr\left\{X_i(t) = 0\right\}, \quad C_i(t) = \Pr\left\{X_i(t) = 1\right\}, \quad 1 \leq i \leq N.
\end{equation}
The \emph{expected state} of the organization at time $t$ is characterized by the vector
\begin{equation}
  \mathbf{C}(t) = (C_1(t), C_2(t), \cdots, C_N(t))^T.
\end{equation}

\subsection{The cyber defense posture}

The cyber defense of an organization against APTs is twofold: \emph{prevention} and \emph{response}. The former aims to protect the secure nodes in the organization's network from compromise, while the latter is devoted to recovering the compromised nodes in the network.

The prevention investment on a node consists of three parts: the cost for purchasing a set of security products for the node, the cost for deploying and configuring the security products, and the cost for enhancing the user's awareness against advanced social engineering attacks. Let $\delta_i$ denote the prevention investment on node $i$. In this work, we assume the prevention investment on each node is linearly proportional to the security level of the node, i.e., $\delta_i = \delta \times w_i$, where the positive constant $\delta$ is referred to as the \emph{prevention coefficient}.

The response investment on a node consists of four parts: the cost for monitoring and analyzing the activities related to the node, the cost for deciding on whether the node is compromised or not, the cost for isolating the node from the network when it is compromised, and the cost for recovering the compromised node. Let $\gamma_i$ denote the response investment on node $i$. In this work, we assume the response investment on each node is linearly proportional to the security level of the node, i.e., $\gamma_i = \gamma \times w_i$, where the positive constant $\gamma$ is referred to as the \emph{response coefficient}.

\subsection{The cyber attack strategy}

The threat of an APT campaign to the organization is twofold: \emph{external attack} and \emph{internal infection}. The former is conducted by the attacker from outside of the network, while the latter is caused by the compromised nodes within the network, both with the same goal of compromising the secure nodes in the network.

Let $B$ denote the budget per unit time for attacking the organization. In this work, we assume $B$ is a constant, which is determined by the attacker prior to the campaign.

Let $x_i$ denote the cost per unit time used for attacking node $i$ when it is secure. In this work, we assume $x_i$ is a constant, which is determined by the attacker prior to the campaign. The \emph{attack strategy} is characterized by the vector
\begin{equation}
  \mathbf{x} = (x_1, x_2, \cdots, x_N).
\end{equation}
Let $||\mathbf{x}||_1$ denote the 1-norm of $\mathbf{x}$, i.e., $||\mathbf{x}||_1 = \sum_{i = 1}^N x_i$. Then, $||\mathbf{x}||_1 = B$. Let $\Omega_B$ denote the admissible set of attack strategies,
\begin{equation}
  \Omega_B = \left\{\mathbf{u} \in \mathbb{R}_+^N : ||\mathbf{u}||_1 = B \right\}.
\end{equation}
Then, we have $\mathbf{x} \in \Omega_B$.

\subsection{A state evolution model of an organization}

For fundamental knowledge on differential dynamical systems, see Ref. \cite{Khalil2002}.

Suppose an APT campaign on an organization starts at time $t = 0$ and terminates at time $t = T$. To model the state evolution of the organization, let us impose a set of hypotheses as follows.

\begin{enumerate}
	
\item [(H$_1$)] Due to external attack and prevention, at any time a secure node $i$ gets compromised at rate $\frac{\alpha x_i}{\delta w_i}$, where the positive constant $\alpha$ is referred to as the \emph{attack coefficient}; which is proportional to the quality of the reconnaissance. This hypothesis is rational, because the rate is (a) proportional to the attack cost per unit time, and (b) inversely proportional to the prevention investment.
	
\item[(H$_2$)] Due to internal infection and prevention, at any time a secure node $i$ gets compromised at the average rate $\frac{\beta\sum_{j=1}^{N}a_{ji}C_j(t)}{\delta w_i}$, where the positive constant $\beta$ is referred to as the \emph{infection coefficient}, which is typically small. Indeed, this coefficient is controllable by the attacker so as to avoid detection. This hypothesis is rational, because the average rate is (a) proportional to the probability of the event that each specific neighboring node is compromised, and (b) inversely proportional to the prevention investment.
	
\item [(H$_3$)] Due to response, at any time a compromised node $i$ gets secure at rate $\gamma w_i$. This hypothesis is rational, because the rate is proportional to the response investment.

This set of hypotheses is schematically shown in Fig. 1.
	
\end{enumerate}

\begin{figure}[H]
	\centering
	\includegraphics[width=0.5\textwidth]{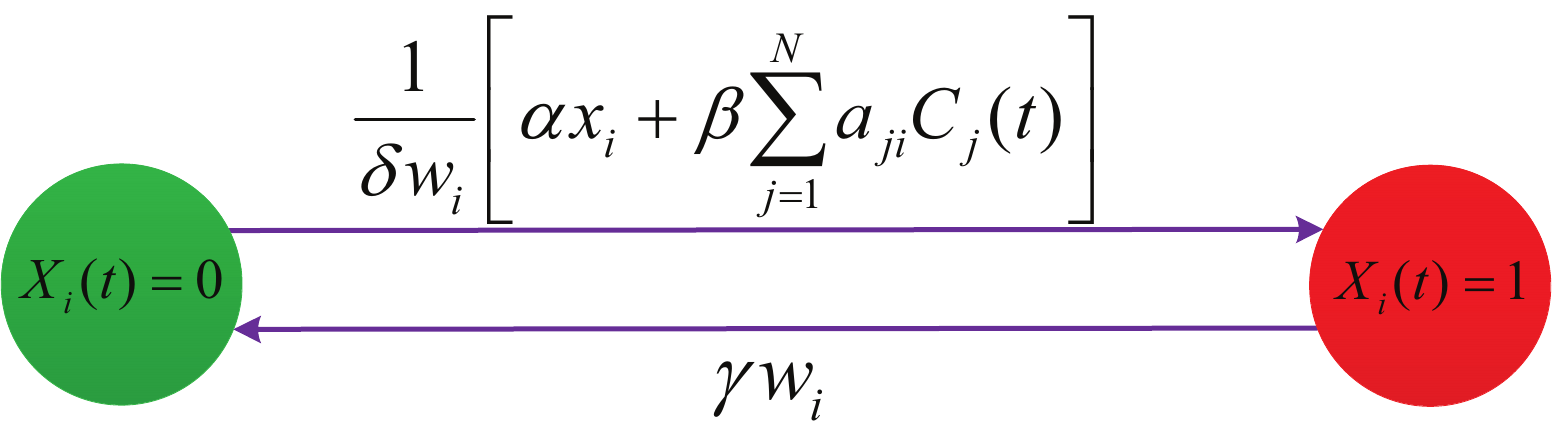}
	\caption{Diagram of hypotheses (H$_1$)-(H$_3$).}
\end{figure}

Based on the above hypotheses, the evolution of the expected state of the organization is modeled as the following differential dynamical system:
\begin{equation}
\frac{dC_i(t)}{dt}= \frac{1}{\delta w_i}\left[\alpha x_i + \beta \sum_{j=1}^{N}a_{ji}C_j(t)\right][1 - C_i(t)] -  \gamma w_iC_i(t), \quad 0 \leq t \leq T, i = 1, \cdots, N.
\end{equation}
We refer to the model as the \emph{Secure-Compromised-Secure} (SCS) model. A SCS model is characterized by the 7-tuple $M_{SCS} = (G, \alpha, \beta, \delta, \gamma, T, \mathbf{x})$.

\subsection{The modeling of the risk assessment problem}

For simplicity, we assume (a) the loss per unit time of an organization owing to a compromised node $i$ is $w_i$, and (b) the profit per unit time of the attacker owing to the compromised node is also $w_i$. This assumption is rational, because the loss and profit both are proportional to the security level of the node.

Consider a SCS model $M_{SCS} = (G, \alpha, \beta, \delta, \gamma, T, \mathbf{x})$. The \emph{expected loss} of the organization caused by implementing the attack strategy $\mathbf{x}$ is
\begin{equation}
  L(\mathbf{x}; G, \alpha, \beta, \delta, \gamma, T) = \int_0^T \sum_{i=1}^Nw_iC_i(t)dt.
\end{equation}

In what follows, we define the risk of an organization as the maximum possible expected loss of the organization over all admissible attack strategies. Let $R(G, \alpha, \beta, \delta, \gamma, T, B)$ denote the risk of the organization,
\begin{equation}
  R(G, \alpha, \beta, \delta, \gamma, T, B) = \max_{\mathbf{x} \in\Omega_B}L(\mathbf{x}; G, \alpha, \beta, \delta, \gamma, T) = \max_{\mathbf{x} \in\Omega_B}\int_0^T \sum_{i=1}^Nw_iC_i(t)dt.
\end{equation}
Obviously, the risk of an organization is dependent upon not only the security posture, $(G, \delta, \gamma)$, but the attack mechanism, $(\alpha, \beta, T, B)$.

Therefore, the original risk assessment problem is modeled as the following constrained optimization problem:
\begin{equation}
\begin{split}
  &\max_{\mathbf{x} \in\Omega_B} L(\mathbf{x}; G, \alpha, \beta, \delta, \gamma, T) = \int_0^T \sum_{i=1}^Nw_iC_i(t)dt, \\
  &\text{s.t.} \quad \frac{dC_i(t)}{dt} = \frac{1}{\delta w_i}\left[\alpha x_i + \beta \sum_{j=1}^{N}a_{ji}C_j(t)\right][1 - C_i(t)] -  \gamma w_iC_i(t), \quad 0 \leq t \leq T,  i = 1, \cdots, N, \\
  & \quad\quad\quad C_i(0) = C_i^*, \quad i = 1, \cdots, N.
 \end{split}
\end{equation}
We refer to the optimization problem as the \emph{risk assessment (RA) model}. A RA model is characterized by the 7-tuple $M_{RA} = (G, \alpha, \beta, \delta, \gamma, T, B)$.

Let $\hat{\mathbf{x}} = (x_1, x_2, \cdots, x_{N-1})$, and let
\begin{equation}
  \hat{\Omega}_B = \left\{\hat{\mathbf{u}} = (u_1, u_2, \cdots, u_{N-1}) \in \mathbb{R}_+^{N-1} : \sum_{i=1}^{N-1} u_i \leq B \right\}.
\end{equation}
Then the RA model (9) can be written in reduced form as follows:\\
\begin{equation}
\begin{split}
&\max_{\hat{\mathbf{x}} \in \hat{\Omega}_B} \hat{L}(\hat{\mathbf{x}}; G, \alpha, \beta, \delta, \gamma, T) = \int_0^T \sum_{i=1}^Nw_iC_i(t)dt, \\
&\text{s.t.} \quad \frac{dC_i(t)}{dt} = \frac{1}{\delta w_i}\left[\alpha x_i + \beta \sum_{j=1}^{N}a_{ji}C_j(t)\right][1 - C_i(t)] -  \gamma w_iC_i(t), \quad 0 \leq t \leq T,  i = 1, \cdots, N - 1, \\
& \quad\quad \frac{dC_N(t)}{dt} = \frac{1}{\delta w_N}\left[\alpha \left(B - \sum_{i = 1}^{N-1}x_i\right) + \beta \sum_{j=1}^{N}a_{jN}C_j(t)\right][1 - C_N(t)] -  \gamma w_NC_N(t), \quad 0 \leq t \leq T, \\
& \quad\quad\quad C_i(0) = C_i^*, \quad i = 1, \cdots, N.
\end{split}
\end{equation}
We refer to the optimization problem as the \emph{reduced risk assessment (RRA) model}. An RRA model is also characterized by the 7-tuple $M_{RRA} = (G, \alpha, \beta, \delta, \gamma, T, B)$. Obviously, we have
\begin{equation}
 R(G, \alpha, \beta, \delta, \gamma, T, B) = \max_{\hat{\mathbf{x}} \in\hat{\Omega}_B} L(\hat{\mathbf{x}}; G, \alpha, \beta, \delta, \gamma, T).
\end{equation}

The RRA model will be used in Section 4.

\section{The influence of different factors on the risk of an organization}

Eq. (8) tells us that the risk of an organization is dependent upon the topology of the network, the fourth coefficients, the attack duration, and the attack budget per unit time. This section is committed to examining the way that these factors affect the organization's risk. For this purpose, the following lemma is needed.

\begin{lm} (Chaplygin Lemma, see Theorem 31.4 in \cite{Szarski1965}) Consider a smooth $n$-dimensional system of differential equations
\begin{equation}
  \frac{d\mathbf{x}(t)}{dt} = \mathbf{f}(\mathbf(\mathbf{x}(t)), \quad t \geq 0,
\end{equation}
and consider the following two systems of differential inequalities:
\begin{equation}
  \frac{d\mathbf{y}(t)}{dt} \leq \mathbf{f}(\mathbf(\mathbf{y}(t)), \quad t \geq 0,
\end{equation}
and
\begin{equation}
\frac{d\mathbf{z}(t)}{dt} \geq \mathbf{f}(\mathbf(\mathbf{z}(t)), \quad t \geq 0,
\end{equation}
where $\mathbf{x}(0) = \mathbf{y}(0) = \mathbf{z}(0)$. Suppose that for any $a_1, \cdots, a_n \geq 0$, there hold
\begin{equation}
  f_i(x_1+a_1, \cdots, x_{i-1}+a_{i-1}, x_i, x_{i+1} + a_{i+1}, \cdots, x_n + a_n) \geq f_i(x_1, \cdots, x_n), \quad i = 1, \cdots, n.
\end{equation}
Then, $\mathbf{y}(t) \leq \mathbf{x}(t)$ and $\mathbf{z}(t) \geq \mathbf{x}(t)$ for all $t \geq 0$.
\end{lm}

\subsection{The influence of the network topology}

The following theorem discloses the influence of the network topology on the risk of an organization.

\begin{thm}
The risk of an organization increases with the addition of new edges to the network.
\end{thm}

\begin{proof}
Consider a pair of RA models, $M_{RA}^{(1)} = (G_1, \alpha, \beta, \delta, \gamma, T, B)$ and $M_{RA}^{(2)} = (G_2, \alpha, \beta, \delta, \gamma, T, B)$, where $G_1$ is a spanning subgraph of $G_2$, i.e., $G_1 = (V, E_1)$, $G_2 = (V, E_2)$, $E_1 \subseteq E_2$. Let $\mathbf{A}(G_1) = \left[a_{ij}^{(1)}\right]_{N \times N}$, $\mathbf{A}(G_2) = \left[a_{ij}^{(2)}\right]_{N \times N}$. Then $a_{ij}^{(1)} \leq a_{ij}^{(2)}$, $1 \leq i, j \leq N$. Let $\left(C_1^{(1)}(t), \cdots, C_N^{(1)}(t)\right)$ be the solution to the SCS model $M_{SCS}^{(1)} = (G_1, \alpha, \beta, \delta, \gamma, T, \mathbf{x})$ with a given initial condition, $\left(C_1^{(2)}(t), \cdots, C_N^{(2)}(t)\right)$ the solution to the SCS model $M_{SCS}^{(2)} = (G_2, \alpha, \beta, \delta, \gamma, T, \mathbf{x})$ with the same initial condition. Then,
\begin{equation}
\frac{dC_i^{(1)}(t)}{dt}= \frac{1}{\delta w_i}\left[\alpha x_i + \beta \sum_{j=1}^{N}a_{ji}^{(1)}C_j^{(1)}(t)\right]\left[1 - C_i^{(1)}(t)\right] -  \gamma w_iC_i^{(1)}(t), \quad 0 \leq t \leq T, i = 1, \cdots, N.
\end{equation}
\begin{equation}
\begin{split}
  \frac{dC_i^{(2)}(t)}{dt} &= \frac{1}{\delta w_i}\left[\alpha x_i + \beta \sum_{j=1}^{N}a_{ji}^{(2)}C_j^{(2)}(t)\right]\left[1 - C_i^{(2)}(t)\right] -  \gamma w_iC_i^{(2)}(t) \\
  &\geq \frac{1}{\delta w_i}\left[\alpha x_i + \beta \sum_{j=1}^{N}a_{ji}^{(1)}C_j^{(2)}(t)\right]\left[1 - C_i^{(2)}(t)\right] -  \gamma w_iC_i^{(2)}(t), \quad 0 \leq t \leq T, i = 1, \cdots, N.
\end{split}
\end{equation}
It follows from Lemma 1 that $C_i^{(1)}(t) \leq C_i^{(2)}(t)$, $0 \leq t \leq T$, $i = 1, 2, \cdots, N$. So,
\begin{equation}
  L(\mathbf{x}; G_1, \alpha, \beta, \delta, \gamma, T) = \int_0^T\sum_{i=1}^Nw_iC_i^{(1)}(t)dt \leq \int_0^T\sum_{i=1}^Nw_iC_i^{(2)}(t)dt = L(\mathbf{x}; G_2, \alpha, \beta, \delta, \gamma, T).
\end{equation}
Hence,
\begin{equation}
\begin{split}
  R(G_1, \alpha, \beta, \delta, \gamma, T, B) &= \max_{\mathbf{x} \in \Omega_B}L(\mathbf{x}; G_1, \alpha, \beta, \delta, \gamma, T)
  \leq \max_{\mathbf{x} \in \Omega_B}L(\mathbf{x}; G_2, \alpha, \beta, \delta, \gamma, T) \\
  &= R(G_2, \alpha, \beta, \delta, \gamma, T, B).
\end{split}
\end{equation}
The proof is complete.
\end{proof}

This theorem implies that the denser the network of an organization, the higher the risk of the organization will be. So, busy business is always accompanied with high risk.

\subsection{The influence of the four coefficients}

The following theorem illuminates the way that the four coefficient in the RA model affects the risk of an organization.

\begin{thm}

The risk of an organization ascends with the attack and infection coefficients, and descends with the prevention and response coefficients.

\end{thm}

The proof of this theorem is analogous to that of Theorem 1 and hence is omitted. The first claim exhibits that a meticulous reconnaissance can enhance the risk of the target organization. The second claim demonstrates that a fast infection can increase the risk. The last two claims show that an increase in security investment always reduces the risk.

\subsection{The influence of the attack duration}

The following theorem reveals the influence of the attack duration on the risk of an organization.

\begin{thm}
The risk of  an organization goes up with the attack duration.
\end{thm}

\begin{proof}
Consider a pair of RA models, $M_{RA}^{(1)} = (G, \alpha, \beta, \delta, \gamma, T_1, B)$ and $M_{RA}^{(2)} = (G, \alpha, \beta, \delta, \gamma, T_2, B)$, where $T_1 < T_2$. Let $\left(C_1^{(1)}(t), \cdots, C_N^{(1)}(t)\right)$ be the solution to the SCS model $M_{SCS}^{(1)} = (G, \alpha, \beta, \delta, \gamma, T_1, \mathbf{x})$ with a given initial condition, $\left(C_1^{(2)}(t), \cdots, C_N^{(2)}(t)\right)$ the solution to the SCS model $M_{SCS}^{(2)} = (G, \alpha, \beta, \delta, \gamma, T_2, \mathbf{x})$ with the same initial condition. Then,
\begin{equation}
  C_i^{(1)}(t) = C_i^{(2)}(t), \quad 0 \leq t \leq T_1, i = 1, \cdots, N.
\end{equation}
So,
\begin{equation}
\begin{split}
  L(\mathbf{x}; G, \alpha, \beta, \delta, \gamma, T_1) &= \int_0^{T_1}\sum_{i=1}^Nw_iC_i^{(1)}(t)dt = \int_0^{T_1}\sum_{i=1}^Nw_iC_i^{(2)}(t)dt
  \leq \int_0^{T_2}\sum_{i=1}^Nw_iC_i^{(2)}(t)dt \\
  &= L(\mathbf{x}; G, \alpha, \beta, \delta, \gamma, T_2).
\end{split}
\end{equation}
Hence,
\begin{equation}
\begin{split}
  R(G, \alpha, \beta, \delta, \gamma, T_1, B) &= \max_{\mathbf{x} \in \Omega_{B}}L(\mathbf{x}; G, \alpha, \beta, \delta, \gamma, T_1)
  \leq \max_{\mathbf{x} \in \Omega_B}L(\mathbf{x}; G, \alpha, \beta, \delta, \gamma, T_2) \\
  &= R(G, \alpha, \beta, \delta, \gamma, T_2, B).
\end{split}
\end{equation}
The proof is complete.
\end{proof}

\subsection{The influence of the attack budget per unit time}

The following theorem demonstrates the influence of the attack budget per unit time on the risk of an organization.

\begin{thm}
The risk of an organization rises with the attack budget per unit time.
\end{thm}

\begin{proof}
	Consider a pair of RA models, $M_{RA}^{(1)} = (G, \alpha, \beta, \delta, \gamma, T, B_1)$ and $M_{RA}^{(2)} = (G, \alpha, \beta, \delta, \gamma, T, B_2)$, where $B_1 < B_2$. Let $\mathbf{A}(G) = \left[a_{ij}\right]_{N \times N}$. For any $\mathbf{x} \in \Omega_{B_1}$, we have $\frac{B_2}{B_1}\mathbf{x} \in \Omega_{B_2}$. Let $\left(C_1^{(1)}(t), \cdots, C_N^{(1)}(t)\right)$ be the solution to the SCS model $M_{SCS}^{(1)} = (G, \alpha, \beta, \delta, \gamma, T, \mathbf{x})$ with a given initial condition, $\left(C_1^{(2)}(t), \cdots, C_N^{(2)}(t)\right)$ the solution to the SCS model $M_{SCS}^{(2)} = (G, \alpha, \beta, \delta, \gamma, T, \frac{B_2}{B_1}\mathbf{x})$ with the same initial condition. Then,
	\begin{equation}
	\frac{dC_i^{(1)}(t)}{dt}= \frac{1}{\delta w_i}\left[\alpha x_i + \beta \sum_{j=1}^{N}a_{ji}C_j^{(1)}(t)\right]\left[1 - C_i^{(1)}(t)\right] -  \gamma w_iC_i^{(1)}(t), \quad 0 \leq t \leq T, i = 1, \cdots, N.
	\end{equation}
	\begin{equation}
	\begin{split}
	\frac{dC_i^{(2)}(t)}{dt} &= \frac{1}{\delta w_i}\left[\alpha\frac{B_2}{B_1} x_i + \beta \sum_{j=1}^{N}a_{ji}C_j^{(2)}(t)\right]\left[1 - C_i^{(2)}(t)\right] -  \gamma w_iC_i^{(2)}(t) \\
	&\geq \frac{1}{\delta w_i}\left[\alpha x_i + \beta \sum_{j=1}^{N}a_{ji}C_j^{(2)}(t)\right]\left[1 - C_i^{(2)}(t)\right] -  \gamma w_iC_i^{(2)}(t), \quad 0 \leq t \leq T, i = 1, \cdots, N.
	\end{split}
	\end{equation}
	It follows from Lemma 1 that $C_i^{(1)}(t) \leq C_i^{(2)}(t)$, $0 \leq t \leq T$, $i = 1, 2, \cdots, N$. Thus,
	\begin{equation}
	L(\mathbf{x}; G, \alpha, \beta, \delta, \gamma, T) = \int_0^T\sum_{i=1}^Nw_iC_i^{(1)}(t)dt \leq \int_0^T\sum_{i=1}^Nw_iC_i^{(2)}(t)dt = L\left(\frac{B_2}{B_1}\mathbf{x}; G, \alpha, \beta, \delta, \gamma, T\right).
	\end{equation}
	Hence,
	\begin{equation}
	\begin{split}
	R(G, \alpha, \beta, \delta, \gamma, T, B_1) &= \max_{\mathbf{x} \in \Omega_{B_1}}L(\mathbf{x}; G, \alpha, \beta, \delta, \gamma, T)
	\leq \max_{\mathbf{x} \in \Omega_{B_1}}L\left(\frac{B_2}{B_1}\mathbf{x}; G, \alpha, \beta, \delta, \gamma, T\right) \\
	& \leq \max_{\mathbf{x} \in \Omega_{B_2}}L(\mathbf{x}; G, \alpha, \beta, \delta, \gamma, T) = R(G, \alpha, \beta, \delta, \gamma, T, B_2).
	\end{split}
	\end{equation}
	The proof is complete.
\end{proof}

\section{An attack strategy}

The RA model characterizing the RA problem has been established in Section 2. We are now confronted with the problem of how to solve the model. As the RA model involves a higher-dimensional nonlinear objective function and a dynamic constraint, it is extremely difficult, if not impossible, to solve the model analytically. In this section, let us turn our attention to the numerical solution of the RA model.

\subsection{The HC attack strategy}

The goal of this subsection is to present a numerical method for solving the RA model. For this purpose, let us first examine the unimodality of the objective function in the RRA model through computer experiments.

\begin{expe}
	
Let $G^{(2)}$ be the connected graph with two nodes labelled 1 and 2. Fig. 2 plots the four functions $\hat{L}(\hat{\mathbf{x}}; G, \alpha, \beta, \delta, \gamma, T)$ with the following combinations of parameters:
\begin{tabbing}
\hspace{5cm} $G$ \quad\quad\= $\alpha$ \quad\quad\= $\beta$ \quad\quad\= $\delta$ \quad\quad\= $\gamma$ \quad\quad\= T \quad\quad\= B\\
\hspace{5cm} $G^{(2)}$ \> 0.5 \> 0.5 \> 0.5/1 \> 1 \> 10 \> 10\\
\hspace{5cm} $G^{(2)}$ \> 0.5/1 \> 1 \> 1 \> 1 \> 10 \> 10
\end{tabbing}
For each of these functions, the sole maximum point is marked in Fig. 4. It is seen that these functions are all unimodal.
	
\end{expe}

\begin{figure}[!t]
	\setlength{\abovecaptionskip}{0.cm}
	\setlength{\belowcaptionskip}{-0.cm}
	\centering
	\includegraphics[width=0.5\textwidth]{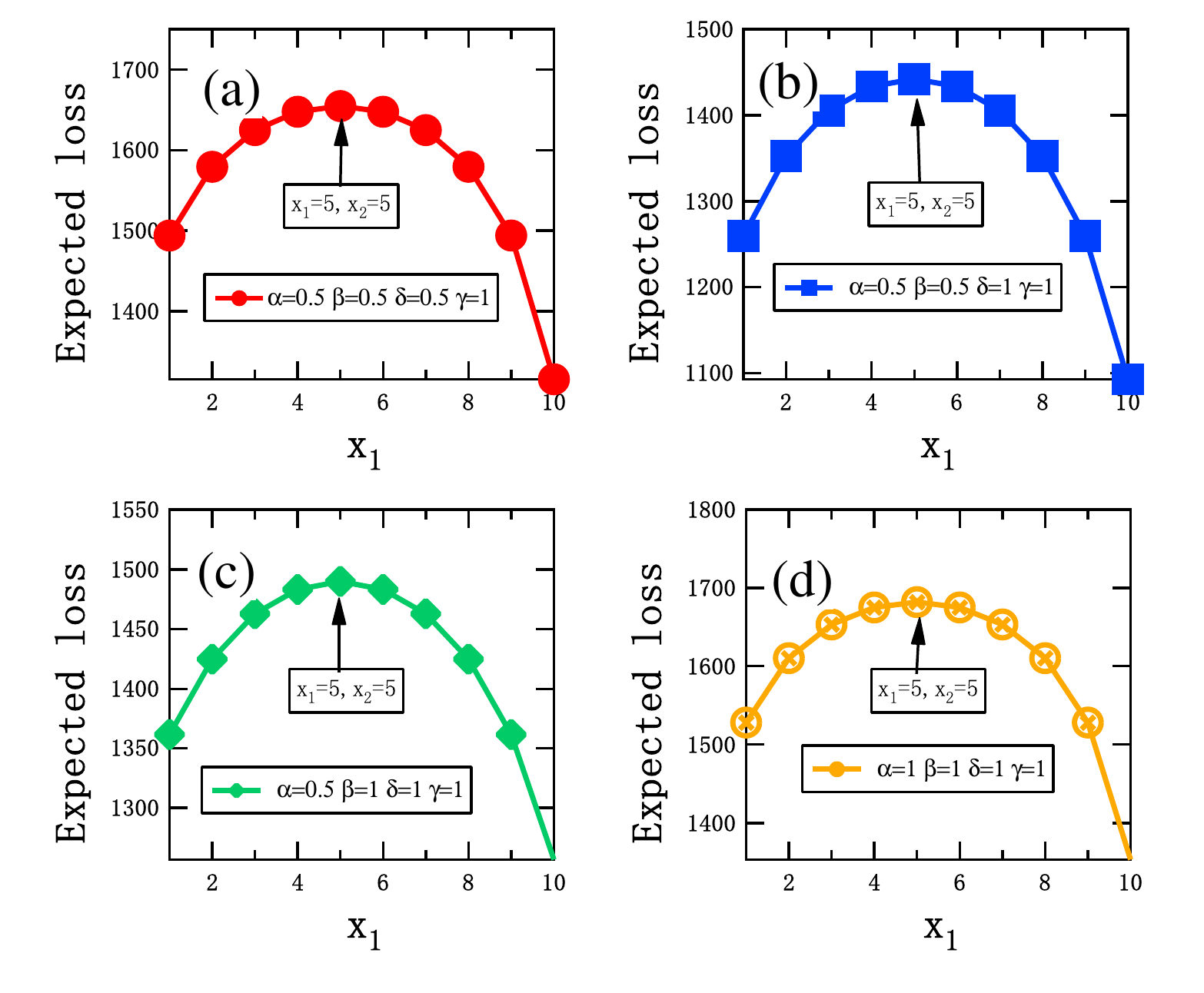}
	\caption{A graphical representation of the objective functions in Experiment 1.}
\end{figure}

\begin{expe}
Up to isomorphism, there are only two different graphs with three nodes, $G_1^{(3)}$ and $G_2^{(3)}$, which are depicted in Fig. 3. Fig. 4 plots the eight functions $\hat{L}(\hat{\mathbf{x}}; G, \alpha, \beta, \delta, \gamma, T)$ with the following combinations of parameters:
\begin{tabbing}
	\hspace{5cm} $G$ \quad\quad\= $\alpha$ \quad\quad\= $\beta$ \quad\quad\= $\delta$ \quad\quad\= $\gamma$ \quad\quad\= T \quad\quad\= B\\
	\hspace{5cm} $G_1^{(3)}$ \> 0.5 \> 0.5 \> 0.5/1 \> 1 \> 10 \> 10\\
	\hspace{5cm} $G_1^{(3)}$ \> 0.5/1 \> 1 \> 1 \> 1 \> 10 \> 10\\
	\hspace{5cm} $G_2^{(3)}$ \> 0.5 \> 0.5 \> 0.5/1 \> 1 \> 10 \> 10\\
	\hspace{5cm} $G_2^{(3)}$ \> 0.5/1 \> 1 \> 1 \> 1 \> 10 \> 10
\end{tabbing}
For each of these functions, the sole maximum point is marked in Fig. 4. It is seen that these functions are all unimodal.

\end{expe}

\begin{figure}[!t]
	\setlength{\abovecaptionskip}{0.cm}
	\setlength{\belowcaptionskip}{-0.cm}
	\centering
	\subfloat[$G_1^{(3)}$]{\includegraphics[width=0.2\textwidth]{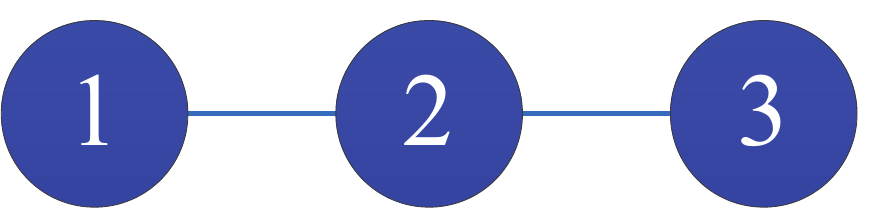}\label{a}}
	\hspace{1cm}
	\subfloat[$G_2^{(3)}$]{\includegraphics[width=0.2\textwidth]{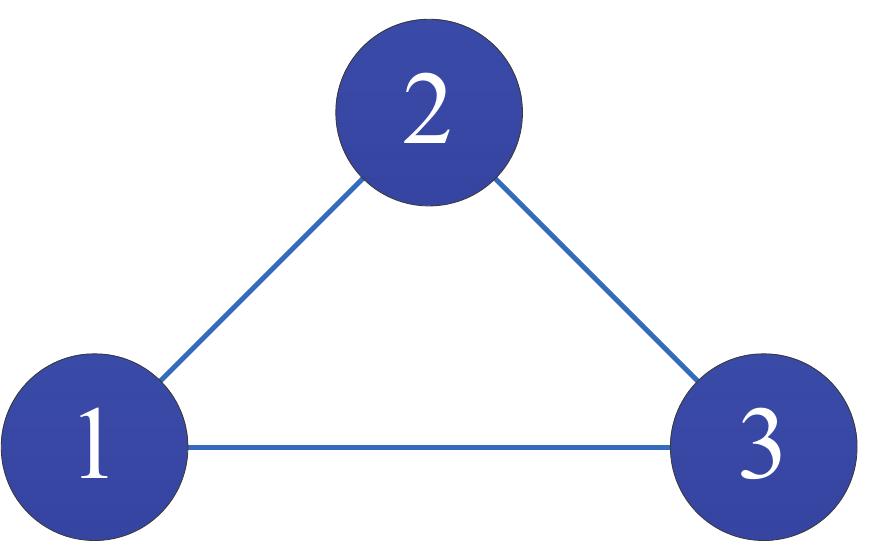}\label{a}}
	\caption{Two connected graphs with three nodes.}
\end{figure}

\begin{figure}[!t]
	\setlength{\abovecaptionskip}{0.cm}
	\setlength{\belowcaptionskip}{-0.cm}
	\centering
	\subfloat[]{\includegraphics[width=0.25\textwidth]{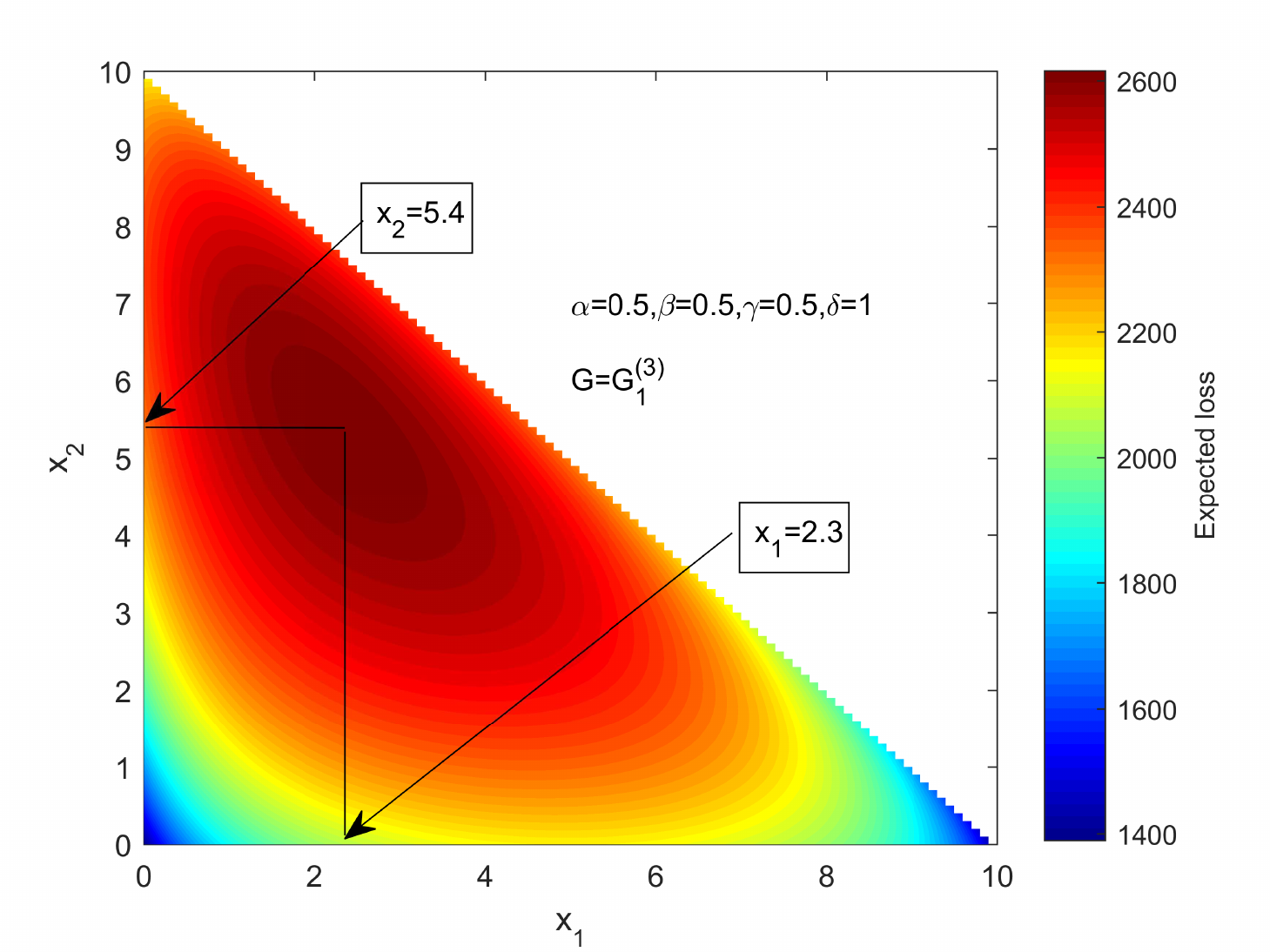}\label{a}}
	\subfloat[]{\includegraphics[width=0.25\textwidth]{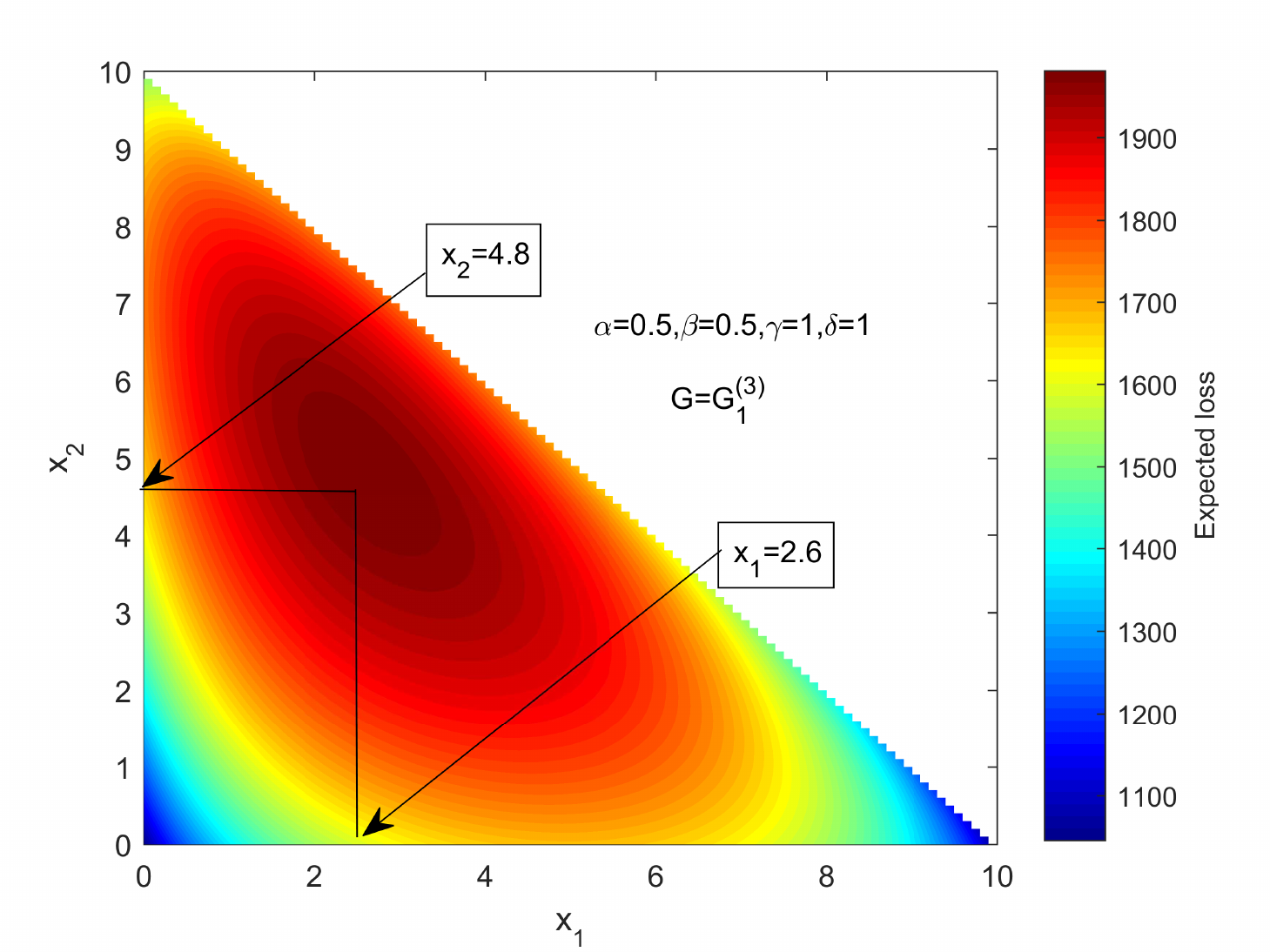}\label{a}}
	\subfloat[]{\includegraphics[width=0.25\textwidth]{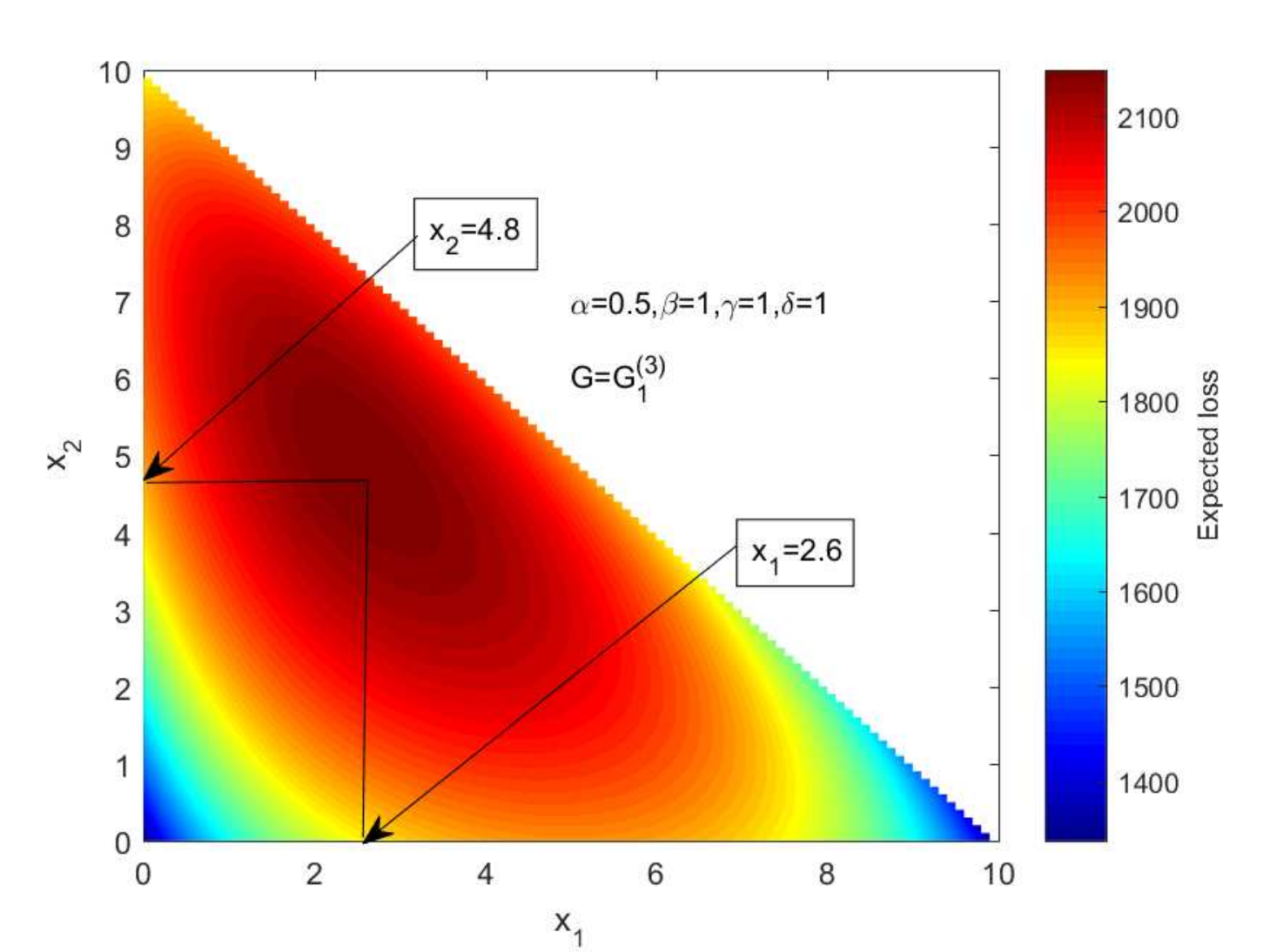}\label{a}}
	\subfloat[]{\includegraphics[width=0.25\textwidth]{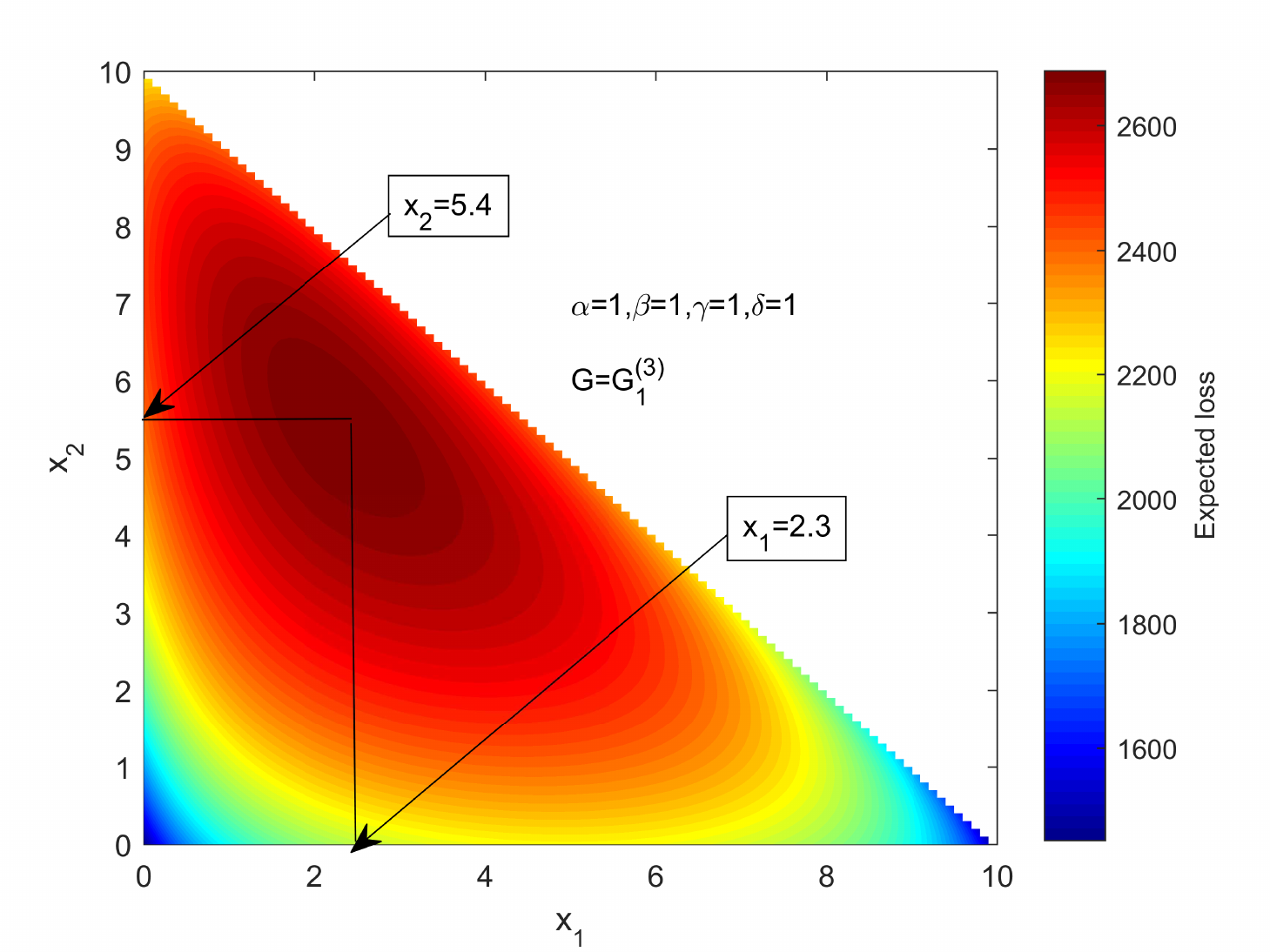}\label{a}} \\
	\subfloat[]{\includegraphics[width=0.25\textwidth]{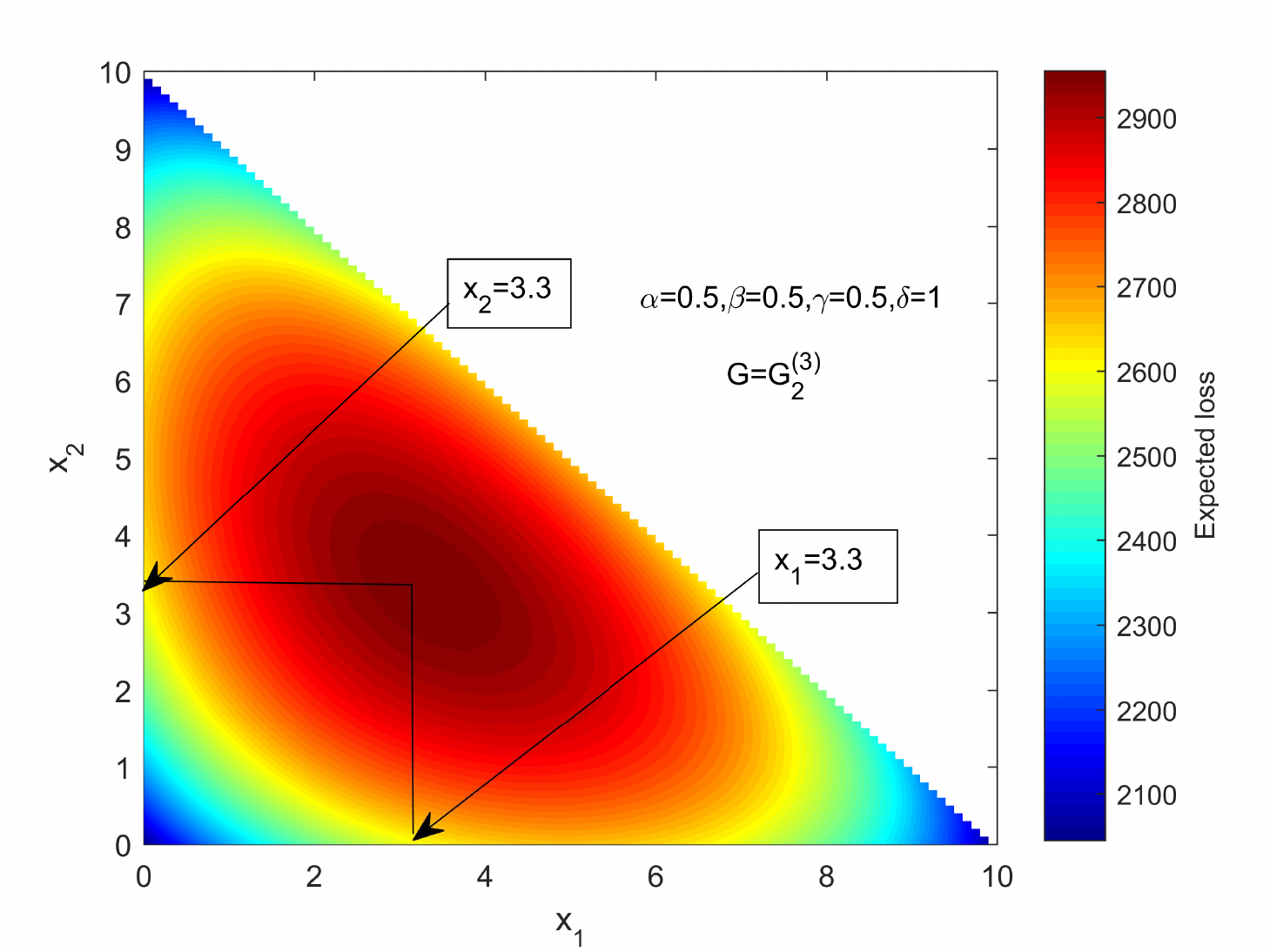}\label{a}}
	\subfloat[]{\includegraphics[width=0.25\textwidth]{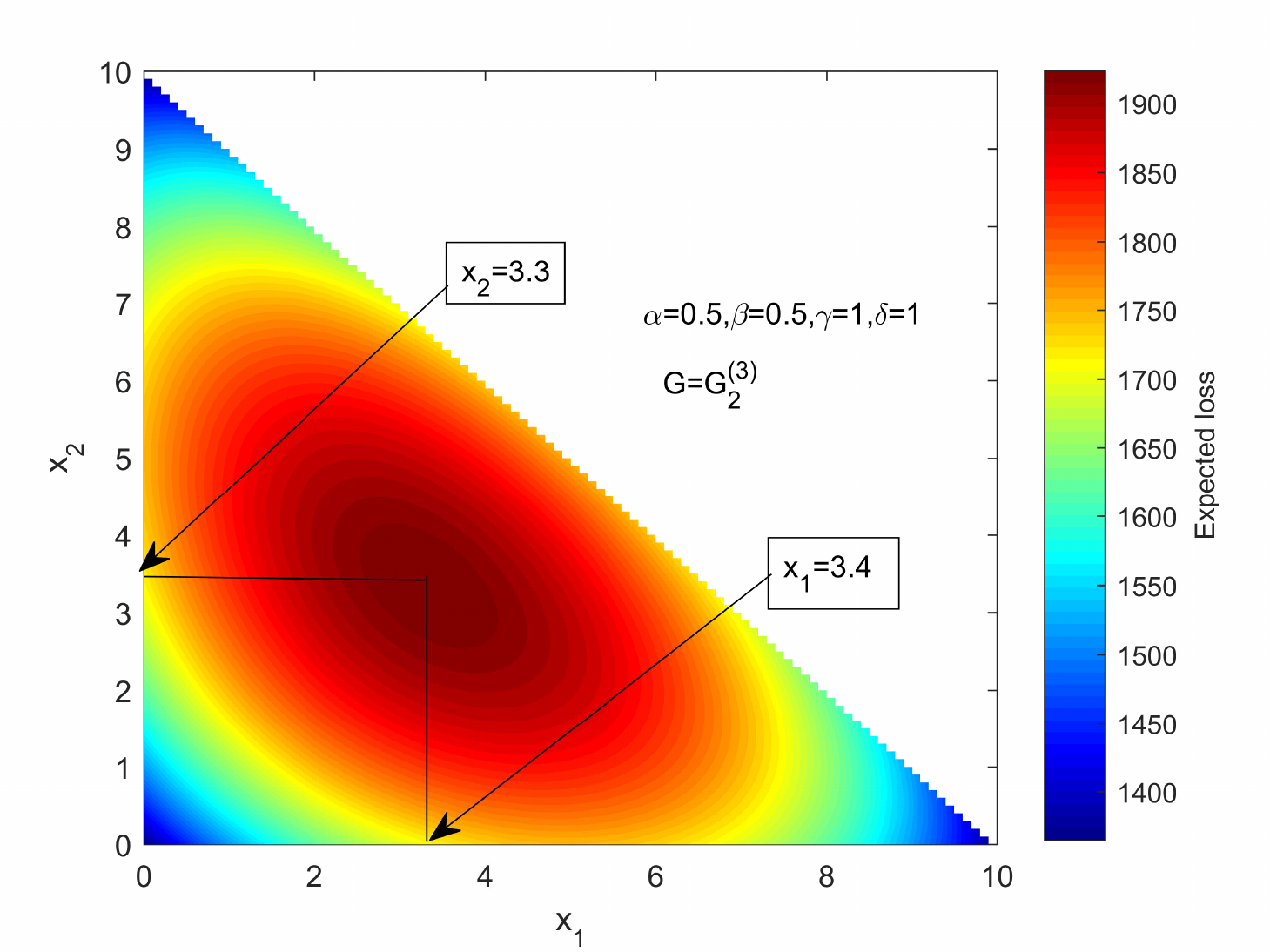}\label{a}}
	\subfloat[]{\includegraphics[width=0.25\textwidth]{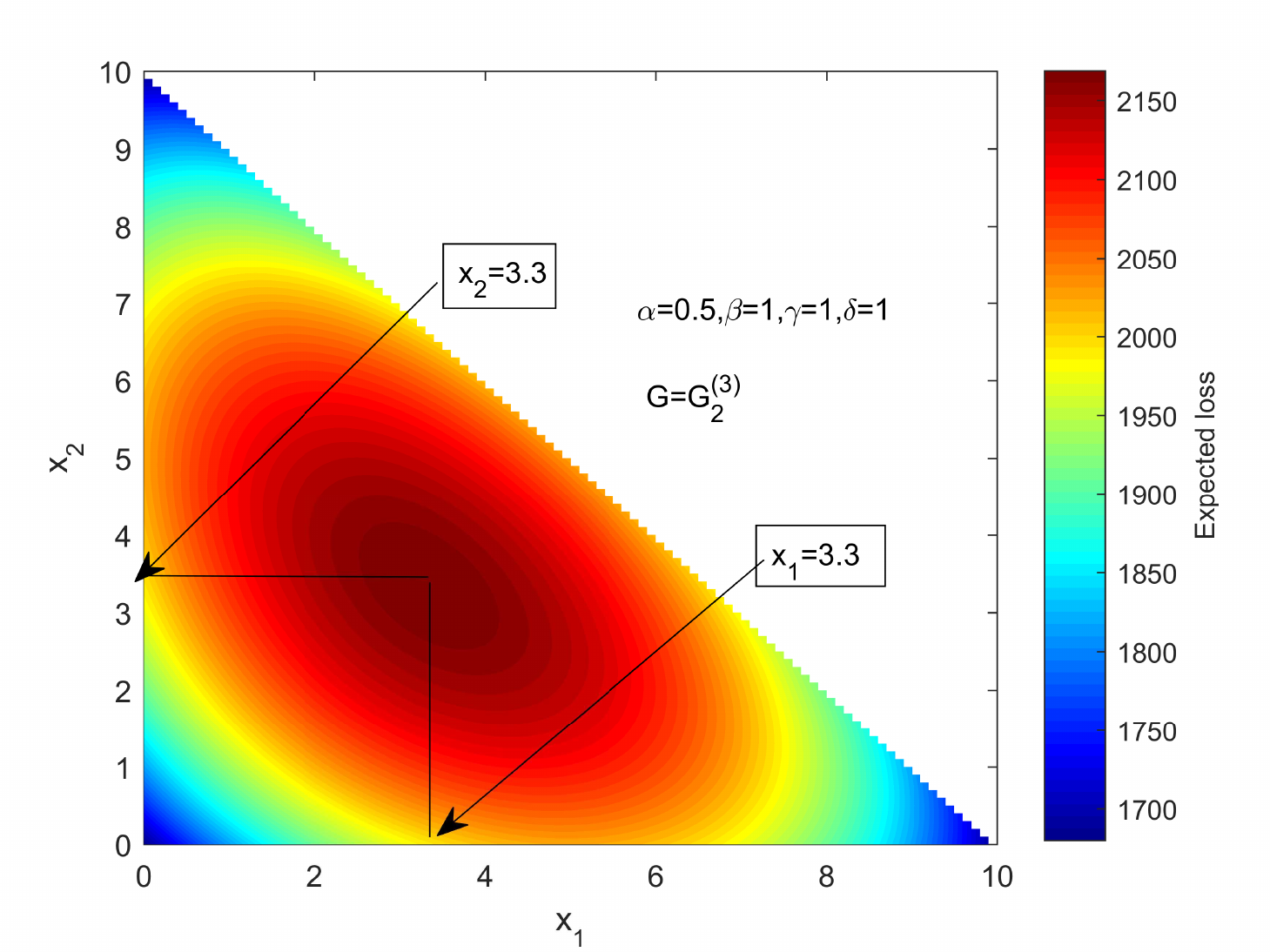}\label{a}}
	\subfloat[]{\includegraphics[width=0.25\textwidth]{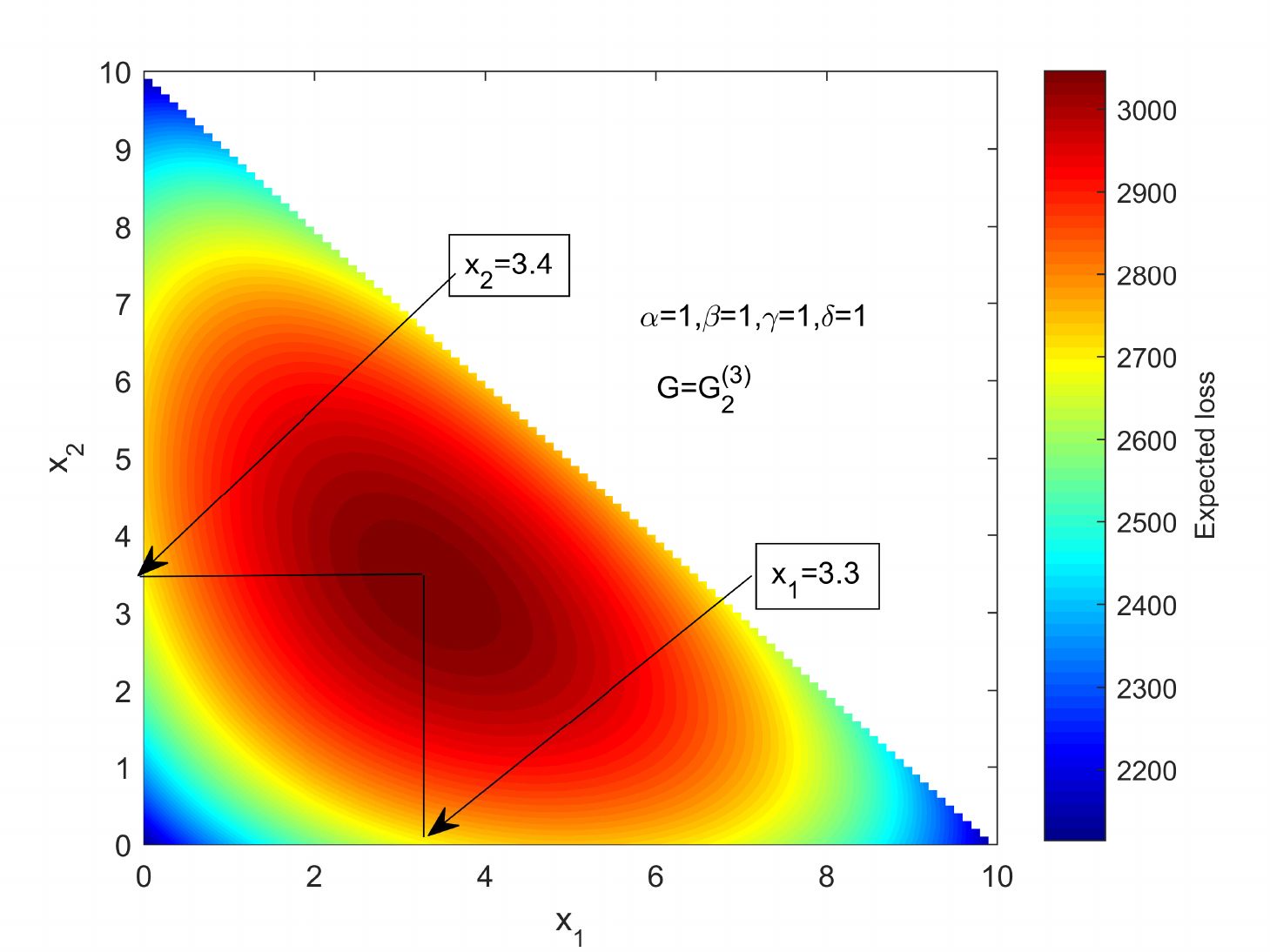}\label{a}}
	\caption{A graphical representation of the objective functions in Experiment 2.}
\end{figure}

\begin{figure}[!t]
	\setlength{\abovecaptionskip}{0.cm}
	\setlength{\belowcaptionskip}{-0.cm}
	\centering
	\subfloat[$G_1^{(4)}$]{\includegraphics[width=0.1\textwidth]{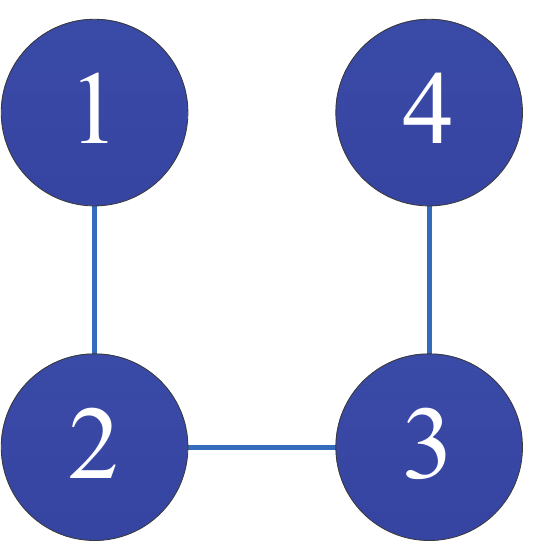}\label{a}}
	\hspace{2ex}
	\subfloat[$G_2^{(4)}$]{\includegraphics[width=0.15\textwidth]{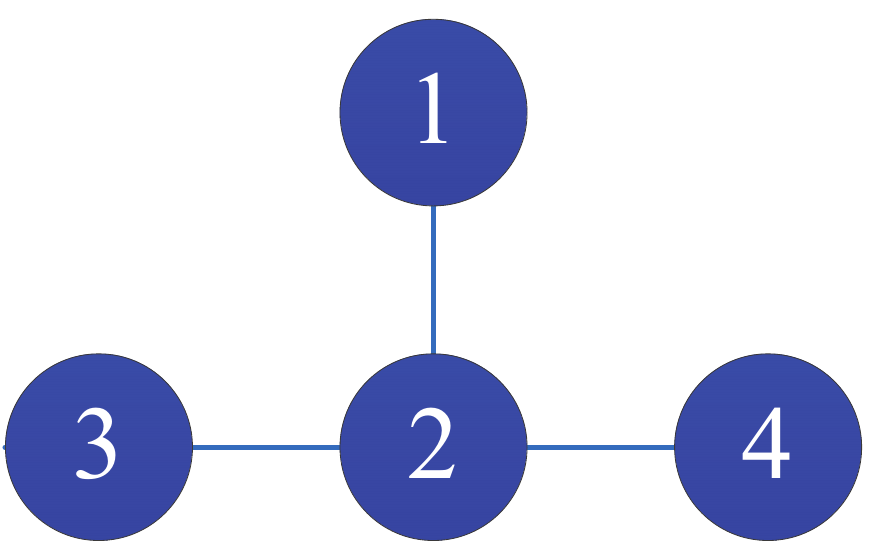}\label{a}}
	\hspace{2ex}
	\subfloat[$G_3^{(4)}$]{\includegraphics[width=0.1\textwidth]{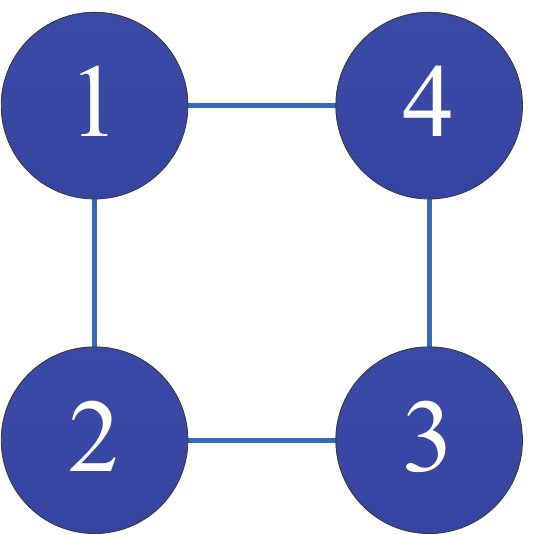}\label{a}}\\
	\subfloat[$G_4^{(4)}$]{\includegraphics[width=0.1\textwidth]{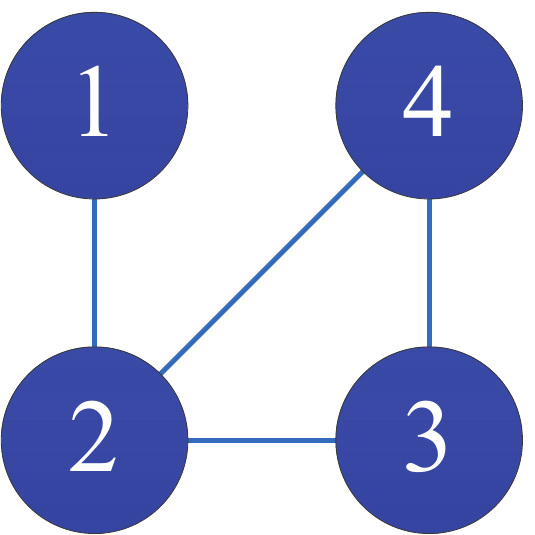}\label{a}}
	\hspace{5ex}
	\subfloat[$G_5^{(4)}$]{\includegraphics[width=0.1\textwidth]{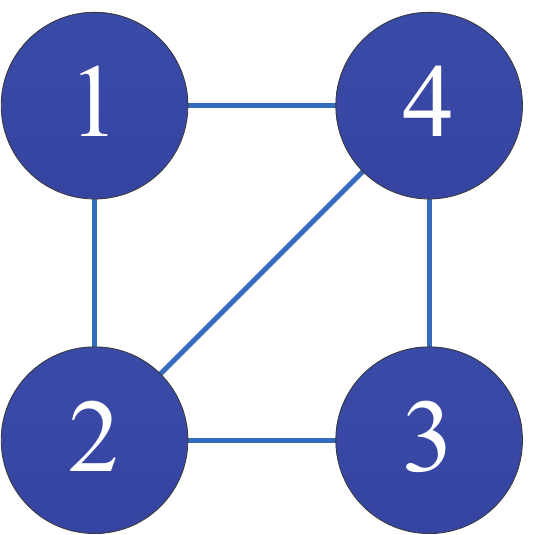}\label{a}}
	\hspace{5ex}
	\subfloat[$G_6^{(4)}$]{\includegraphics[width=0.1\textwidth]{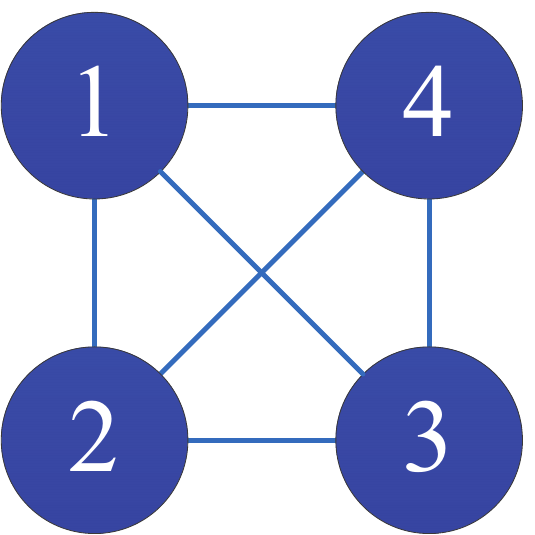}\label{a}}
	\caption{Six connected graphs with four nodes.}
\end{figure}

\begin{expe} Up to isomorphism, there are totally six different graphs with four nodes, $G_i^{(4)}$, $i = 1, \cdots, 6$, which are shown in Fig. 5. Fig. 6 plots a cross figure for each of the 12 functions $\hat{L}(\hat{\mathbf{x}}; G, \alpha, \beta, \delta, \gamma, T)$ with the following combinations of parameters:
\begin{tabbing}
	\hspace{5cm} $G$ \quad\quad\= $\alpha$ \quad\quad\= $\beta$ \quad\quad\= $\delta$ \quad\quad\= $\gamma$ \quad\quad\= T \quad\quad\= B\\
	\hspace{5cm} $G_1^{(4)}$ \> 0.5 \> 0.5 \> 0.5/1 \> 1 \> 10 \> 10\\
	\hspace{5cm} $G_2^{(4)}$ \> 0.5 \> 0.5 \> 0.5/1 \> 1 \> 10 \> 10\\
	\hspace{5cm} $G_3^{(4)}$ \> 0.5 \> 0.5 \> 0.5/1 \> 1 \> 10 \> 10\\
	\hspace{5cm} $G_4^{(4)}$ \> 0.5 \> 0.5 \> 0.5/1 \> 1 \> 10 \> 10\\
	\hspace{5cm} $G_5^{(4)}$ \> 0.5 \> 0.5 \> 0.5/1 \> 1 \> 10 \> 10\\
	\hspace{5cm} $G_6^{(4)}$ \> 0.5 \> 0.5 \> 0.5/1 \> 1 \> 10 \> 10
\end{tabbing}
It is seen that these cross functions are all unimodal. More extensive experiments demonstrate that the 12 original functions are all unimodal.
	
\end{expe}

\begin{figure}[!t]
	\setlength{\abovecaptionskip}{0.cm}
	\setlength{\belowcaptionskip}{-0.cm}
	\centering
	\subfloat[]{\includegraphics[width=0.25\textwidth]{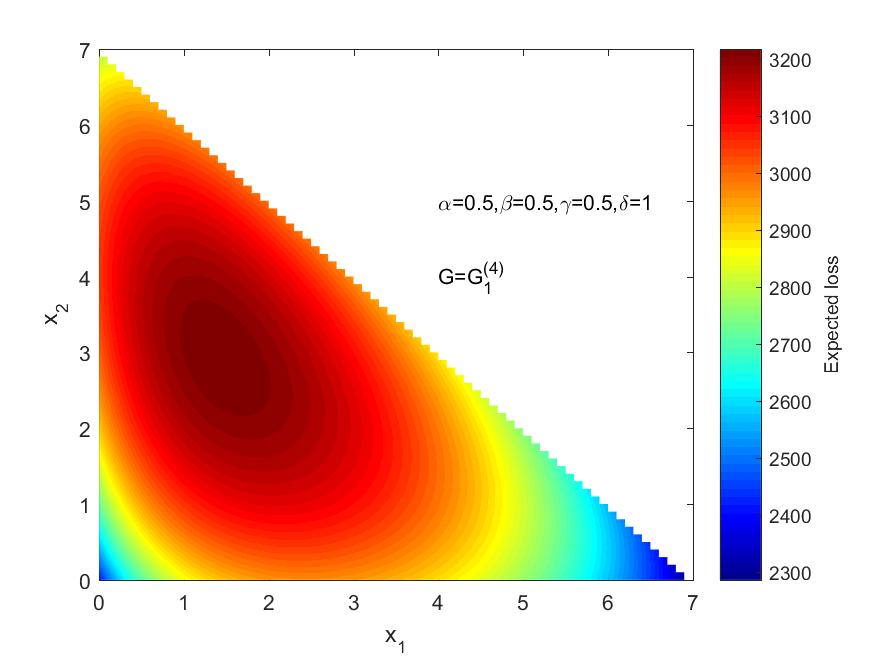}\label{a}}
	\subfloat[]{\includegraphics[width=0.25\textwidth]{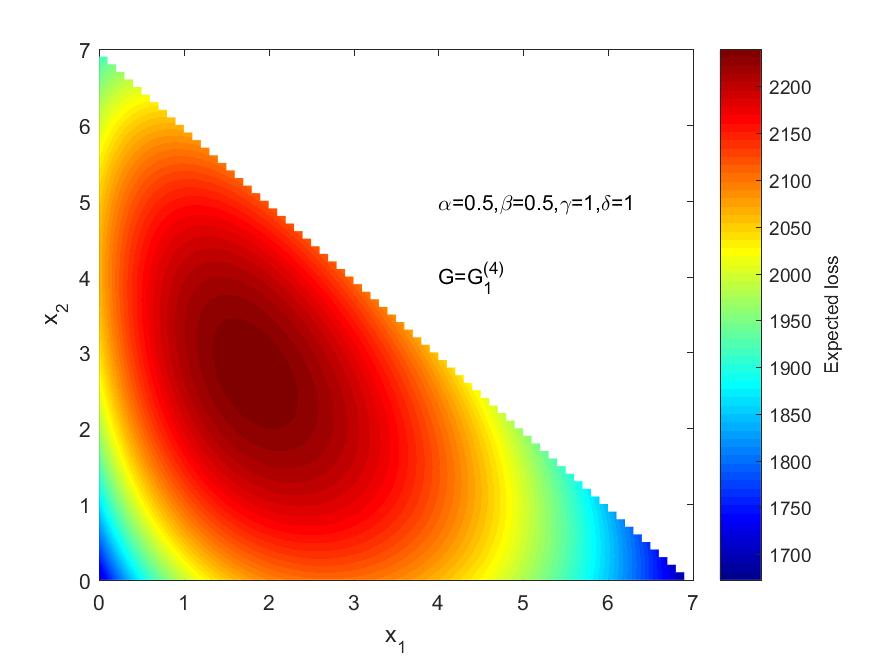}\label{a}}
	\subfloat[]{\includegraphics[width=0.25\textwidth]{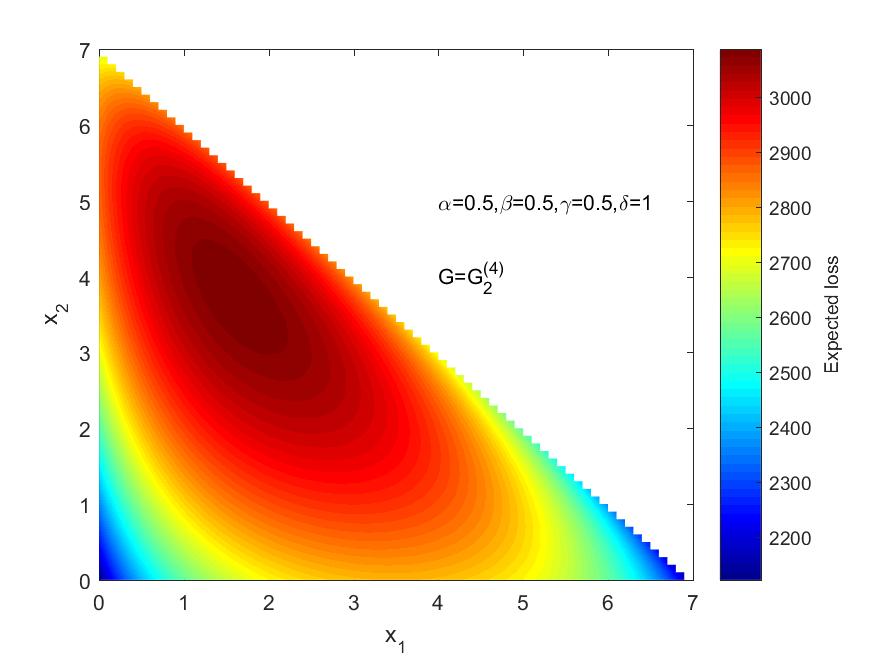}\label{a}}
	\subfloat[]{\includegraphics[width=0.25\textwidth]{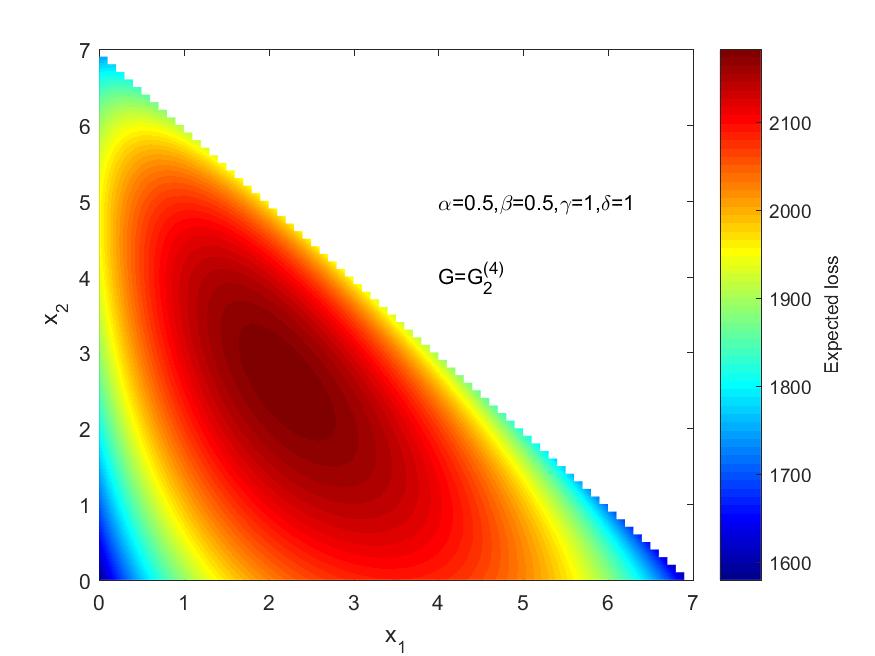}\label{a}} \\
	\subfloat[]{\includegraphics[width=0.25\textwidth]{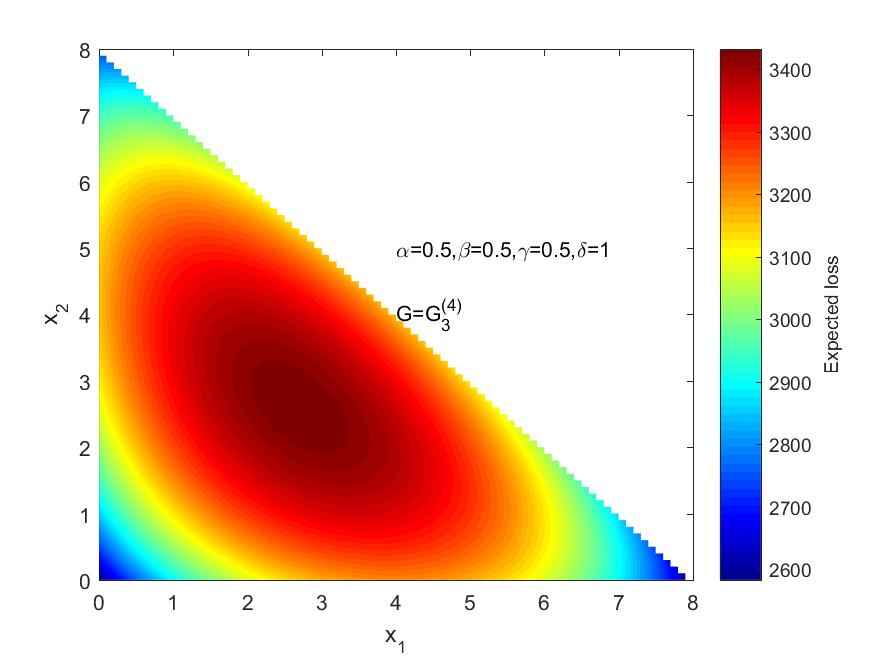}\label{a}}
	\subfloat[]{\includegraphics[width=0.25\textwidth]{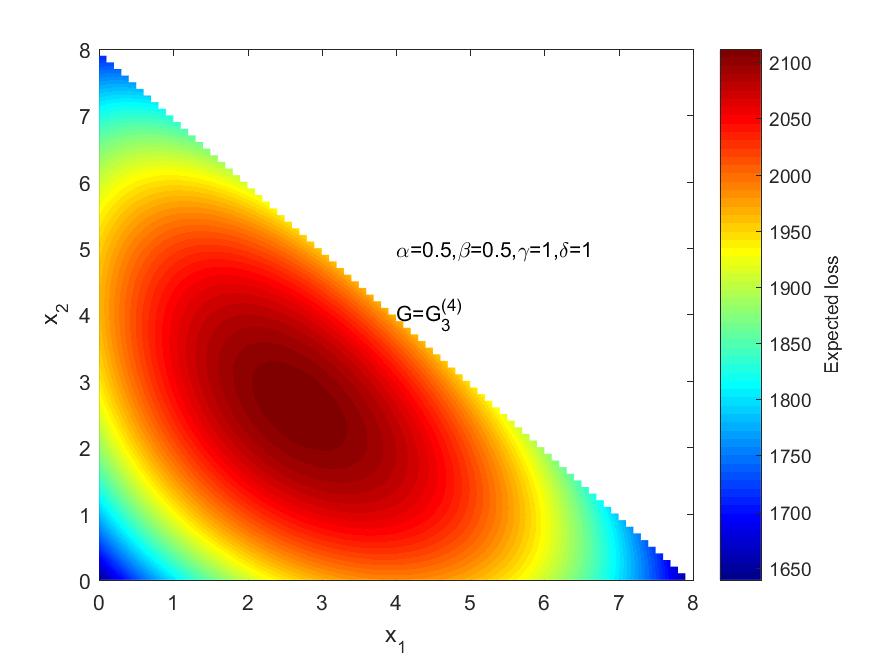}\label{a}}
	\subfloat[]{\includegraphics[width=0.25\textwidth]{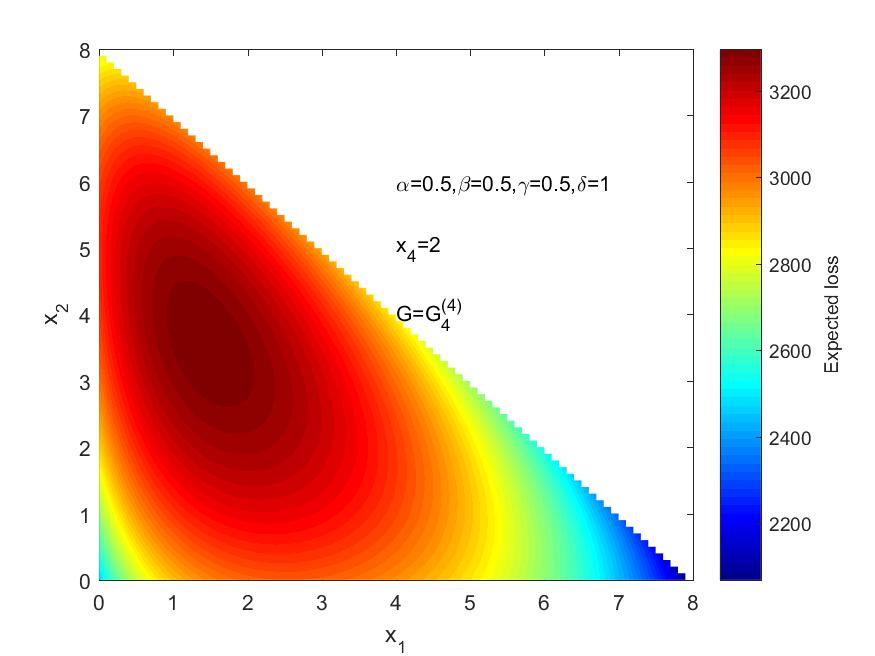}\label{a}}
	\subfloat[]{\includegraphics[width=0.25\textwidth]{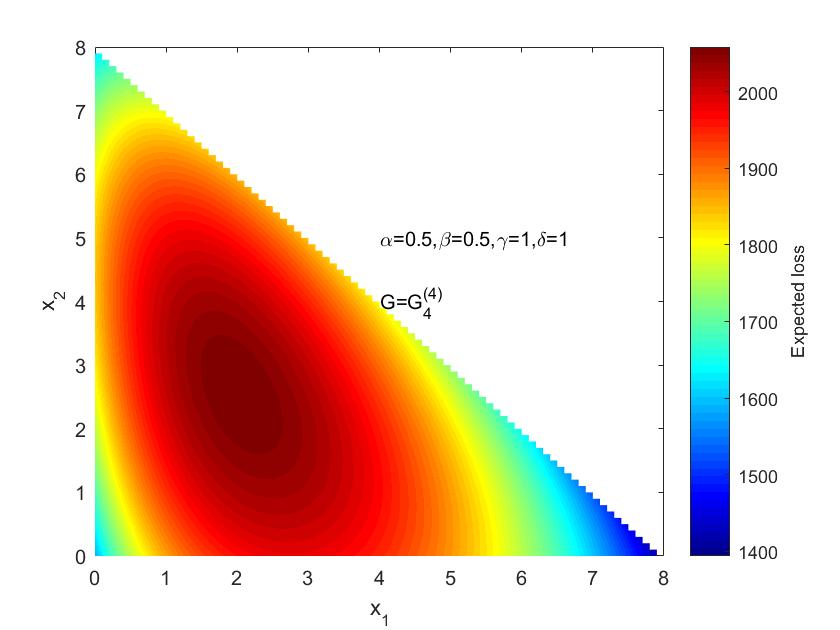}\label{a}}  \\
	\subfloat[]{\includegraphics[width=0.25\textwidth]{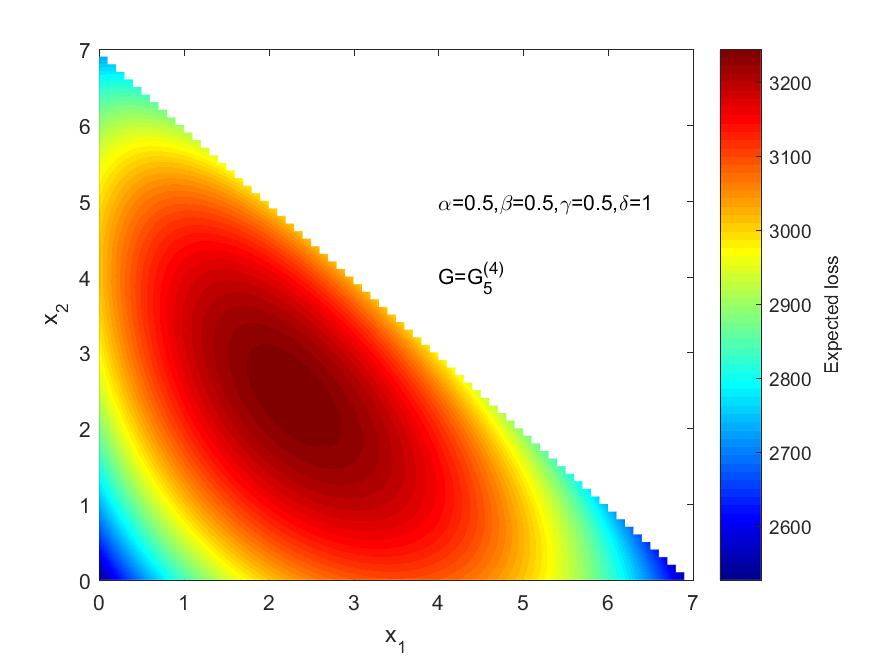}\label{a}}
	\subfloat[]{\includegraphics[width=0.25\textwidth]{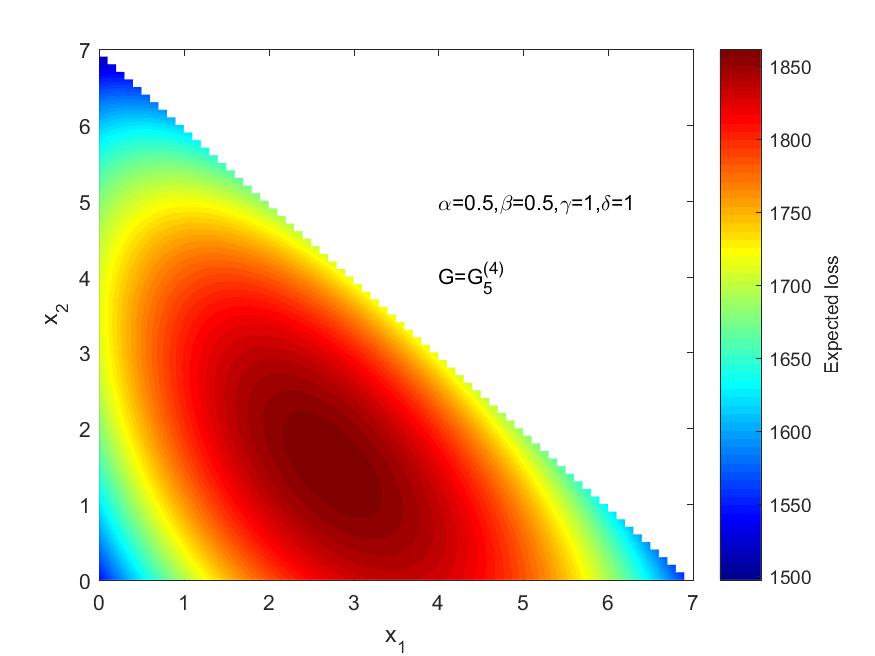}\label{a}}
	\subfloat[]{\includegraphics[width=0.25\textwidth]{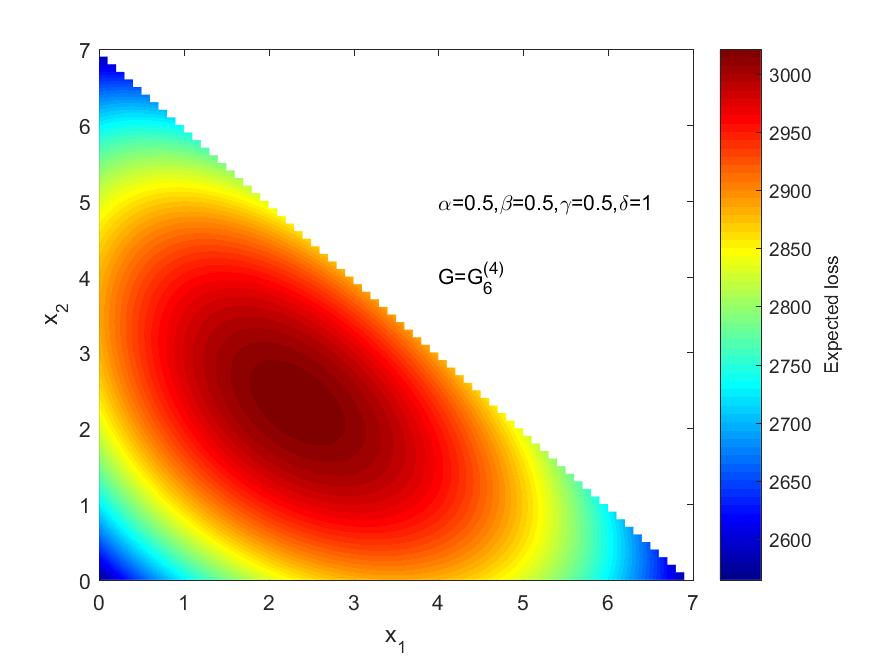}\label{a}}
	\subfloat[]{\includegraphics[width=0.25\textwidth]{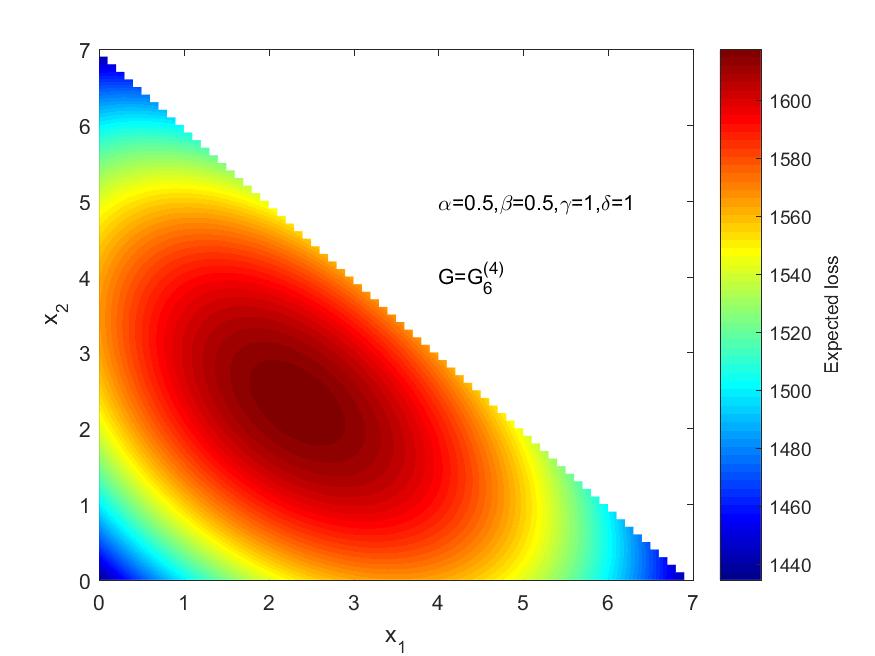}\label{a}}
	\caption{Some cross figures of the functions in Experiment 3.}
\end{figure}

We conclude from these and many similar experiments that the objective function in each RRA model is unimodal, which in turn implies that the objective function in each RA model is unimodal. Hence, for each RA model, the solution obtained through hill climbing is much likely to be optimal.

To formulate our hill-climbing method, we need to introduce a notion as follows.

Let $M_{RA} = (G, \alpha, \beta, \delta, \gamma, T, B)$, $\epsilon > 0$ a small number. The \emph{$\epsilon$-neighborhood} of $\mathbf{x} \in \Omega_B$, denoted $N_{\epsilon}(\mathbf{x})$, is defined as
\begin{equation}
N_{\epsilon}(\mathbf{x}) = \{\mathbf{y} \in \Omega_B : \mathbf{y} - \mathbf{x} \text{ has exactly two nonzero components, one being } \epsilon, \text{ the other -} \epsilon \}.
\end{equation}
And every $\mathbf{y} \in N_{\epsilon}(\mathbf{x})$ is referred to as a \emph{$\epsilon$-neighbor} of $\mathbf{x}$.

Now, we are ready to formulate our method for finding an attack strategy.

\begin{algorithm}[H]
	\caption{HILL-CLIMBING}
	\label{alg1}
	\hspace*{0.02in} {\bf Input} $M_{RA} = (G, \alpha, \beta, \delta, \gamma, T, B)$; $\epsilon = 10^{-6}$.\\
	\hspace*{0.02in} {\bf Output} $\mathbf{x} \in \Omega_B$; $L(\mathbf{x}; G, \alpha, \beta, \delta, \gamma, T)$.
	\begin{algorithmic}[1]
		\STATE randomly choose $\mathbf{x} \in \Omega_B$;
		\WHILE {$\mathbf{x}$ has a $\epsilon$-neighbor $\mathbf{y}$ such that $L(\mathbf{y}; G, \alpha, \beta, \delta, \gamma, T) > L(\mathbf{x}; G, \alpha, \beta, \delta, \gamma, T)$}
		\STATE $\mathbf{x} := \mathbf{y}$;
		\ENDWHILE
		\STATE return $(\mathbf{x}, L(\mathbf{x}; G, \alpha, \beta, \delta, \gamma, T))$.
	\end{algorithmic}
\end{algorithm}

We refer to the attack strategy obtained by executing the HILL-CLIMBING algorithm on a RA model as the \emph{HC strategy} for the RA model, and the expected loss owing to the HC strategy as the \emph{HC risk} for the RA model.

It is seen from Experiments 1-3 that, for each of these RA models, the associated HC strategy is optimal. Through extensive computer experiments, we conclude the following result.

\emph{For each and every RA model, the associated HC strategy is optimal.}

\subsection{Five heuristic attack strategies}

To examine the optimality of the HC strategy, we need to make comparisons on larger networks between this strategy and some other attack strategies. For this purpose, below let us formulate five heuristic attack strategies.

The first heuristic attack strategy is to use up the attack budget to attack a single node of the highest security level. That is,
\begin{equation}
  \mathbf{x} = (0, \cdots, 0, B, 0, \cdots, 0),
\end{equation}
where the target node is of the highest security level, with the ice being broken arbitrarily. We refer to the attack strategy as the \emph{highest security-level (HS) strategy}.

The second heuristic attack strategy is to deplete the attack budget to attack a single node of the lowest security level. That is,
\[
\mathbf{x} = (0, \cdots, 0, B, 0, \cdots, 0),
\]
where the target node is of the lowest security level, with the deadlock being broken arbitrarily. We refer to the attack strategy as the \emph{lowest security-level (LS) strategy}.

The third heuristic attack strategy is to assign to each node an attack cost that is linearly proportional to the security level of the node. That is,
\begin{equation}
\mathbf{x} = (\frac{Bw_1}{\sum_{i=1}^Nw_i}, \frac{Bw_2}{\sum_{i=1}^Nw_i}, \cdots, \frac{Bw_N}{\sum_{i=1}^Nw_i}).
\end{equation}
We refer to the attack strategy as the \emph{security-level first (SF) strategy}.

The fourth heuristic attack strategy is to assign to each node an attack cost that is inversely linearly proportional to the security level of the node. That is,
\begin{equation}
\mathbf{x} = (\frac{\frac{B}{w_1}}{\sum_{i=1}^N\frac{1}{w_i}}, \frac{\frac{B}{w_2}}{\sum_{i=1}^N\frac{1}{w_i}}, \cdots, \frac{\frac{B}{w_N}}{\sum_{i=1}^N\frac{1}{w_i}}).
\end{equation}
We refer to the attack strategy as the \emph{security-level last (SL) strategy}.

The fifth heuristic attack strategy is to allocate the attack budget uniformly among all nodes. That is
\begin{equation}
\mathbf{x} = (\frac{B}{N}, \frac{B}{N}, \cdots, \frac{B}{N}).
\end{equation}
We refer to the attack strategy as the \emph{uniform (UN) strategy}.

\subsection{Comparative experiments}

This section conducts experimental comparisons between the HC strategy and the five heuristic attack strategies in terms of the expected loss. For this purpose, let us describe three networks that will be used in the following experiments.

Small-world networks are networks that are generated by randomly rewiring some edges of regular networks. Fig. 7 plots a small-world network with 50 nodes, which is obtained by executing the algorithm proposed by Watts and Strogatz \cite{Watts1998}. Let $G_{SW}$ denote this network.

\begin{figure}[!t]
	\setlength{\abovecaptionskip}{0.cm}
	\setlength{\belowcaptionskip}{-0.cm}
	\centering
	\includegraphics[width=0.5\textwidth]{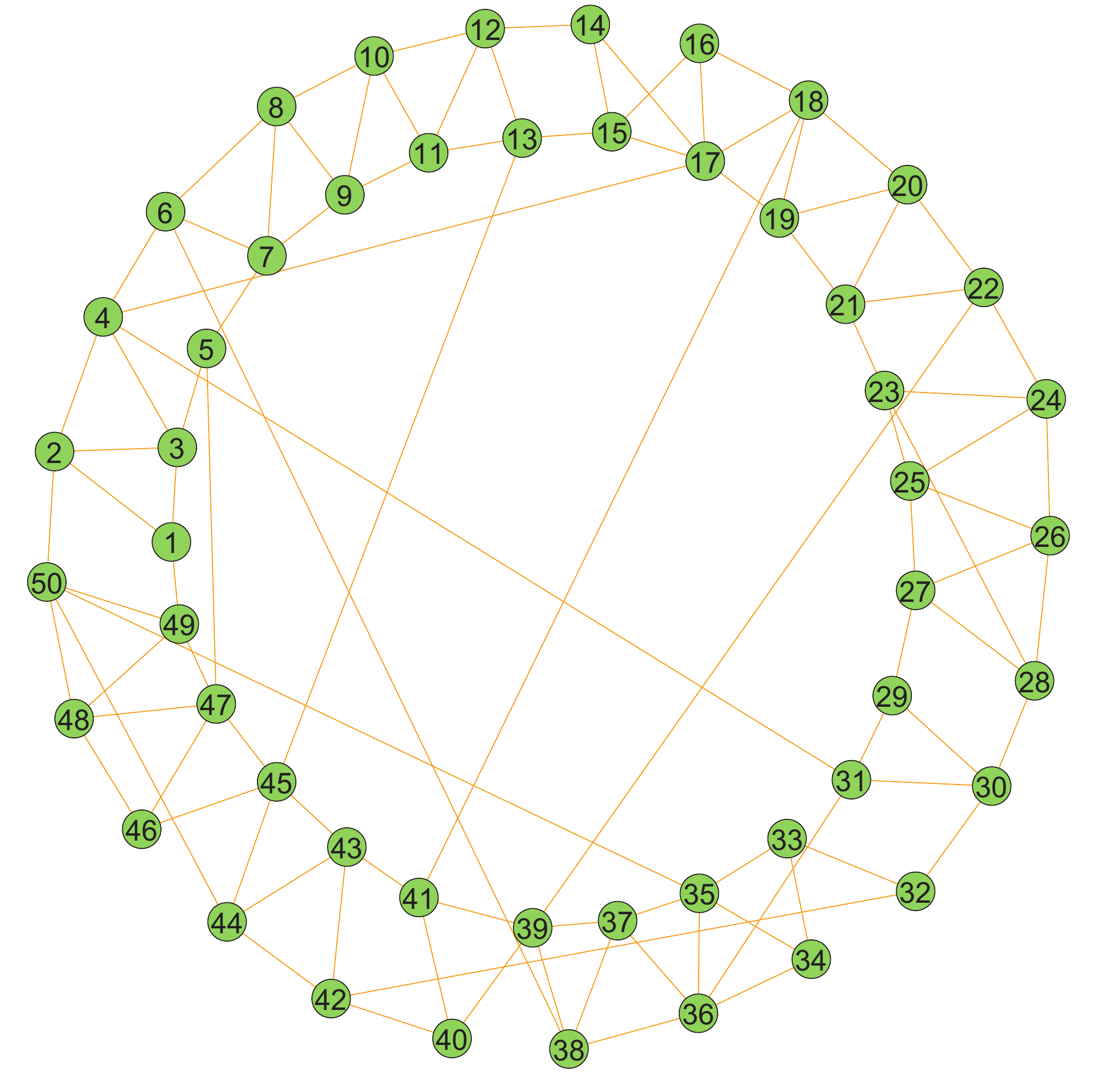}
	\caption{The small-world network $G_{SW}$.}
\end{figure}

Scale-free networks are networks with an approximate power-law degree distribution. Fig. 8 depicts a scale-free network with 50 nodes, which is obtained by executing the algorithm proposed by Barabasi and Albert \cite{Barabasi1999}. Let $G_{SF}$ denote this network.

\begin{figure}[!t]
	\setlength{\abovecaptionskip}{0.cm}
	\setlength{\belowcaptionskip}{-0.cm}
	\centering
	\includegraphics[width=0.5\textwidth]{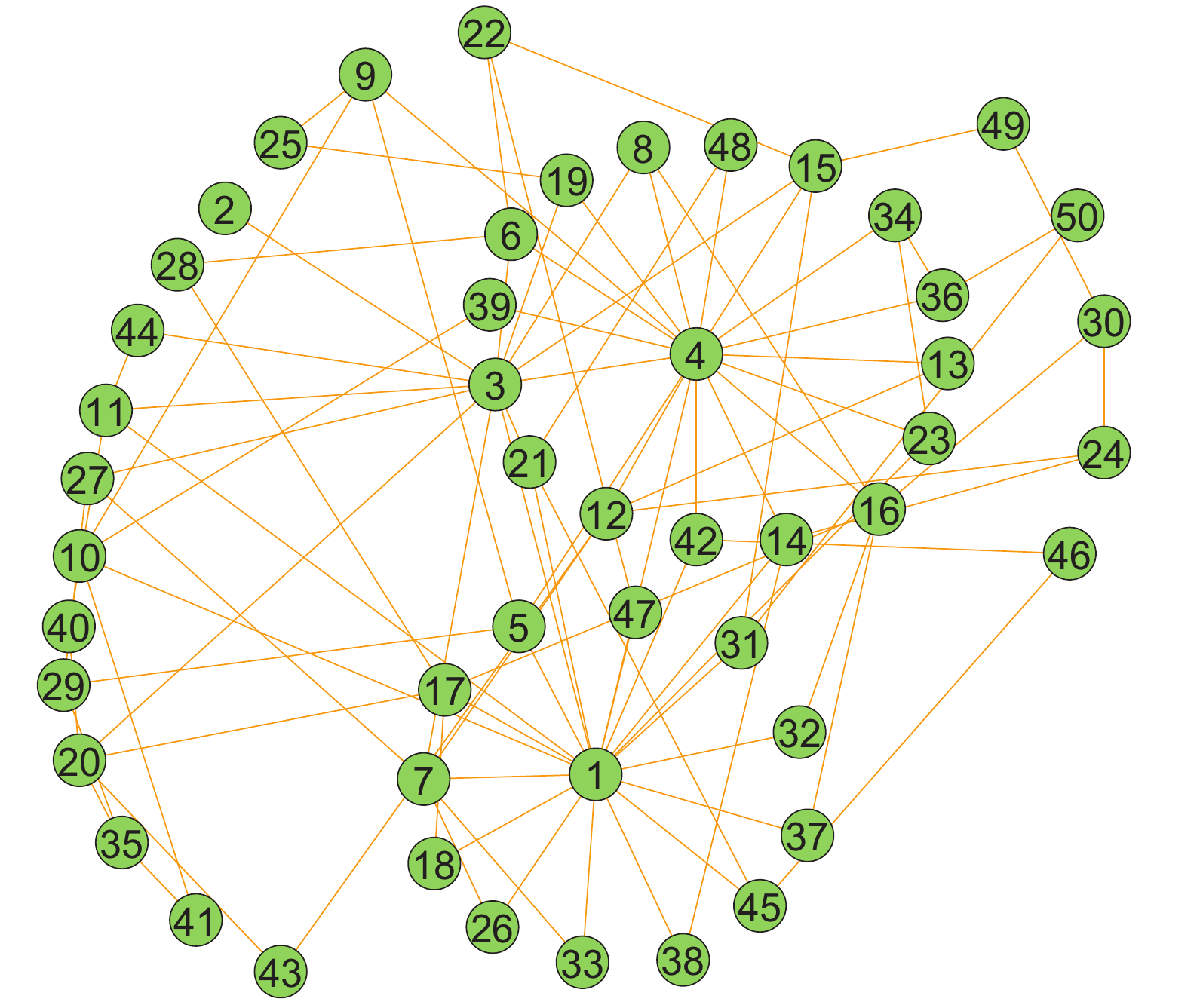}
	\caption{The scale-free network $G_{SF}$.}
\end{figure}

Fig. 9 exhibits a realistic network with 49 nodes, which comes from Ref. \cite{konect}. Let $G_{KO}$ denote this network.

\begin{figure}[!t]
	\setlength{\abovecaptionskip}{0.cm}
	\setlength{\belowcaptionskip}{-0.cm}
	\centering
	\includegraphics[width=0.5\textwidth]{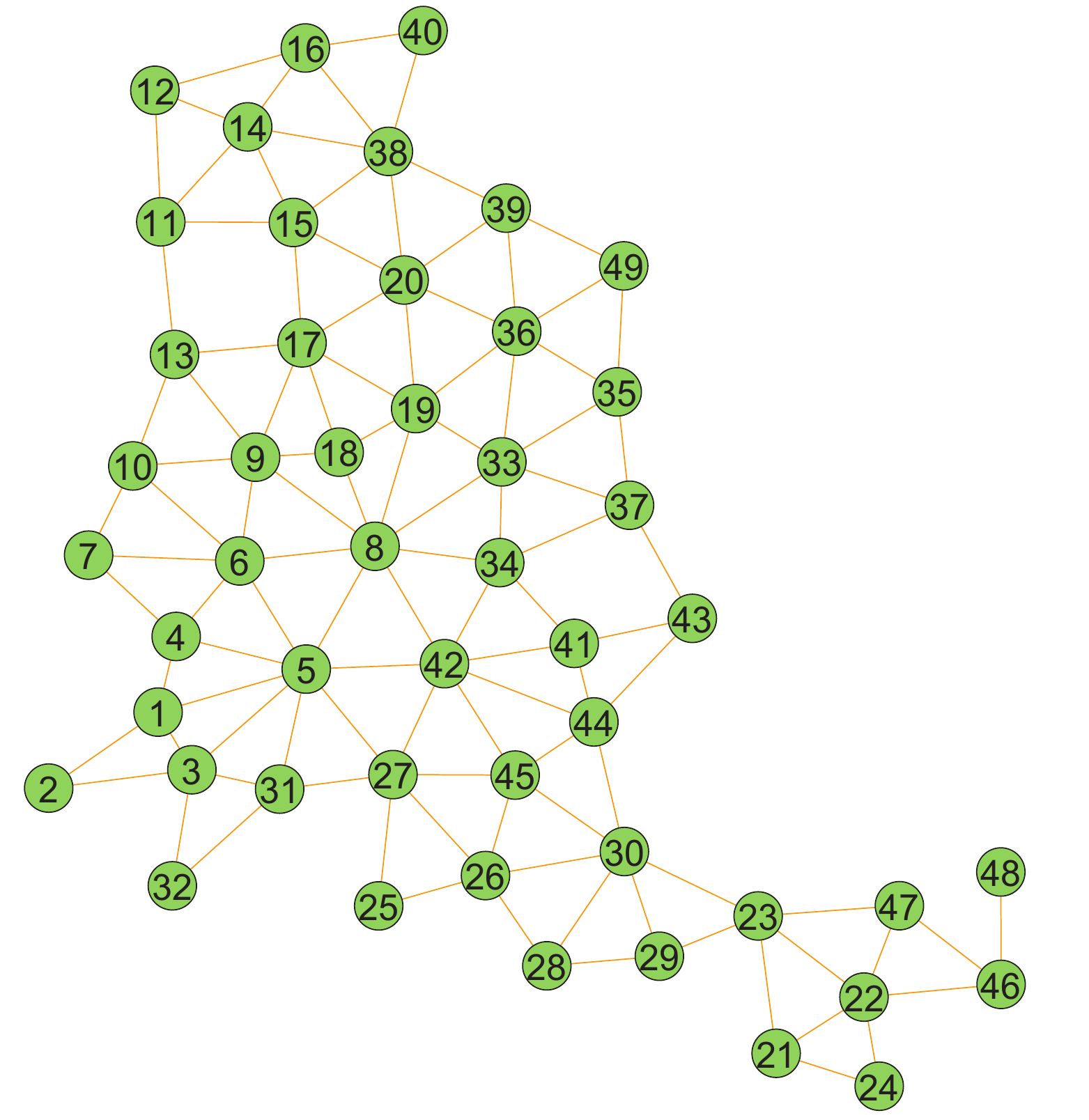}
	\caption{The realistic network $G_{US}$.}
\end{figure}

\begin{expe}
Consider a set of RA models $(G, 1, 1, 1, 1, T, B)$, where $G \in \{G_{SW}, G_{SF}, G_{US}\}$, either (a) $T = 5$ and $B \in \{1, 2, \dots, 10\}$, or (b) $B = 10$ and $T \in \{5, 6, \cdots, 15\}$. For each of these RA models, the HC strategy is compared with the five heuristic attack strategies in terms of the expected loss, and the experimental results are all shown in Fig. 10. It is seen that, for all these RA models, the HC strategy outperforms the five heuristic strategies.
\end{expe}

\begin{figure}[!t]
	\setlength{\abovecaptionskip}{0.cm}
	\setlength{\belowcaptionskip}{-0.cm}
	\centering
	\includegraphics[width=0.8\textwidth]{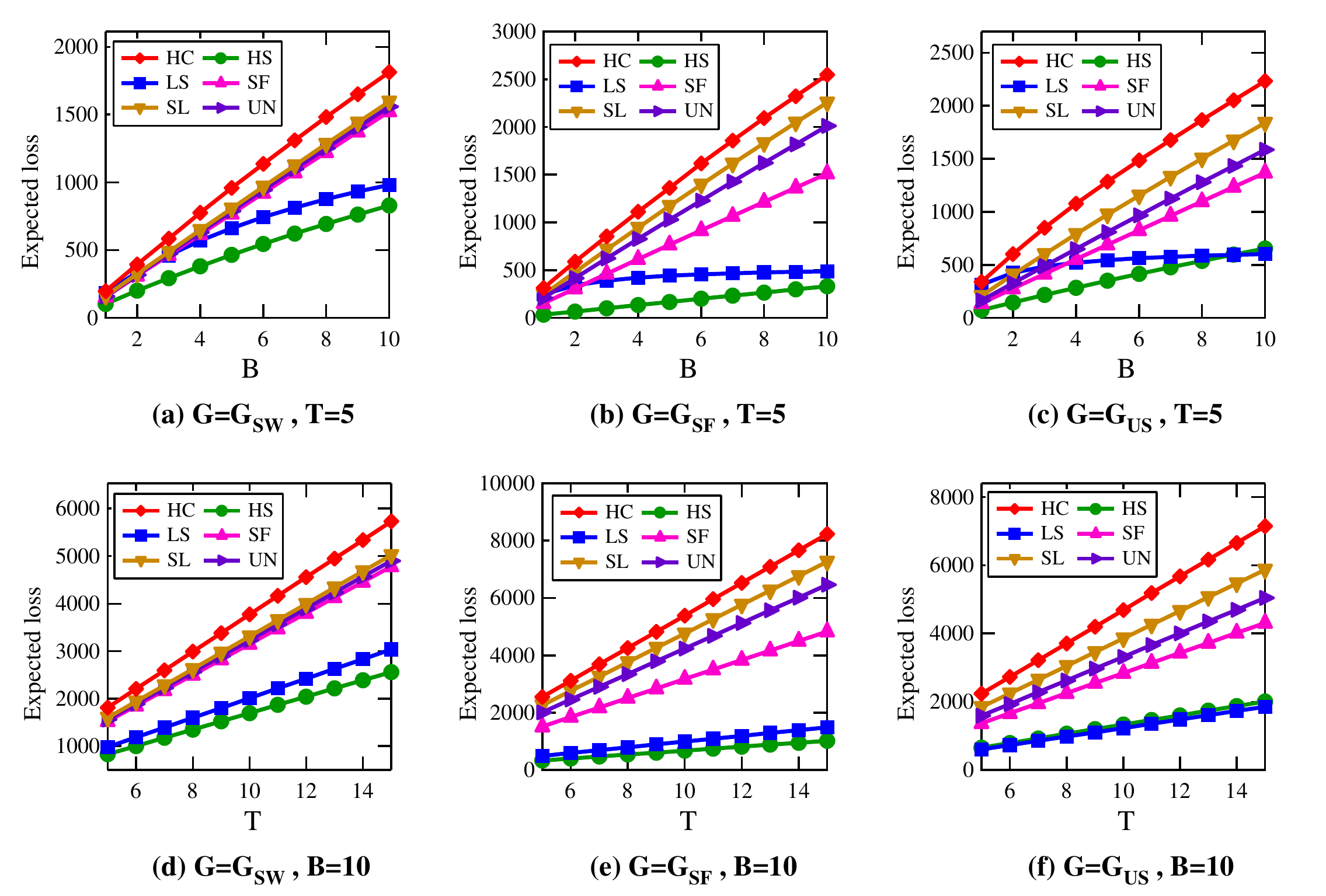}
	\caption{The results in Experiment 4.}
\end{figure}

Based on extensive experiments, we conclude that the HC strategy for each RA model is optimal, i.e., the HC risk is exactly the maximum possible expected loss. This implies that the HC strategy is the biggest threat to an organization, and the HC risk is an indicator of the organization's risk.

\section{Further discussions}

Consider a RA model $M_{RA} = (G, \alpha, \beta, \delta, \gamma, T, B)$. For an attack strategy $\mathbf{x}$, the expected cost benefit of the attacker is
\begin{equation}
\frac{L(\mathbf{x})}{BT} = \frac{1}{||\mathbf{x}||_1}\cdot \frac{1}{T} \int_0^T \sum_{i=1}^Nw_iC_i(t)dt.
\end{equation}

We refer to the expected cost benefit associated with the HC strategy as the attacker's \emph{HC cost benefit}. Based on the results given in the previous section, the HC cost benefit is much likely to be the highest cost benefit an attacker can achieve. Therefore, both the attacker and defender should be concerned with the influence of the attack budget per unit time and the attack duration on the HC cost benefit. This section examines these influences.

\subsection{The influence of the attack budget per unit time}

First, let us examine the influence of the attack budget per unit time on the HC cost benefit.

\begin{expe}
Consider a set of RA models $M_{RA} = (G, 1, 0.5, 1, 0.5, T, B)$, where $G \in \{G_{SW}, G_{SF}, G_{US}\}$, $T \in \{5, 10, 15\}$, and $B \in \{1, 2, \cdots, 10\}$. For each of the RA models, the HC cost benefit is plotted in Fig. 11. It is seen that the HC cost benefit drops with the attack budget per unit time.
\end{expe}

Extensive experiments exhibit similar phenomena. Hence, we conclude that the HC cost benefit always declines with the attack budget per unit time. Therefore, the power of APTs is limited in terms of the HC cost benefit. This sounds a good news for organizations under APTs.

\begin{figure}[H]
	\setlength{\abovecaptionskip}{0.cm}
	\setlength{\belowcaptionskip}{-0.cm}
	\centering
	\includegraphics[width=0.8\textwidth]{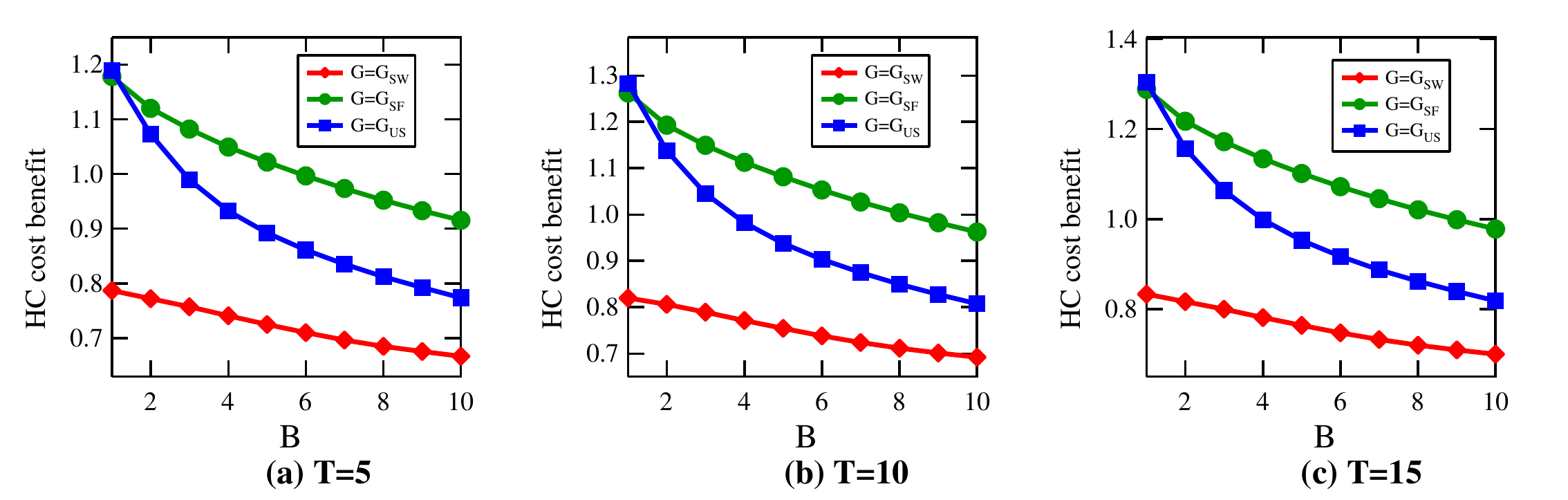}
	\caption{The results in Experiment 5.}
\end{figure}

\subsection{The influence of the attack duration}

Second, let us examine the influence of the attack duration on the HC cost benefit.

\begin{expe}
Consider a set of RA models  $(G, 1, 0.5, 0.5, 1, T, B)$, where $G \in \{G_{SW}, G_{SF}, G_{US}\}$, $B \in \{5, 10, 15\}$, and $T \in \{1, 2, \cdots, 10\}$. For each of the RA models, the HC cost benefit is plotted in Fig. 12. It is seen that, with the extension of the attack duration, the HC cost benefit goes up but flattens out quickly.
\end{expe}

Extensive experiments exhibit similar phenomena. This result demonstrates that, although a short-term APT can achieve a significant increment in HC cost benefit, this increment would become inappreciable with the prolonged attack duration. This conclusion is a good news for organizations, because the motive to conduct an extended APT campaign recedes.

\begin{figure}[H]
	\setlength{\abovecaptionskip}{0.cm}
	\setlength{\belowcaptionskip}{-0.cm}
	\centering
	\includegraphics[width=0.8\textwidth]{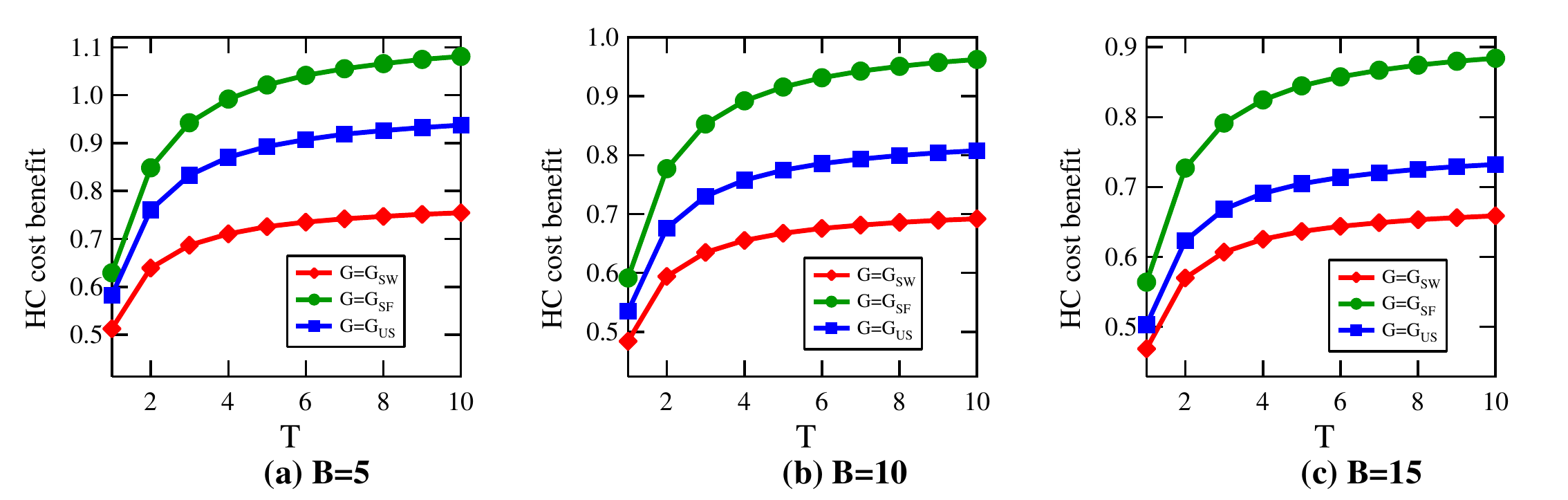}
	\caption{The results in Experiment 6.}
\end{figure}



%\bibliographystyle{IEEEtran}

\section{Concluding remarks}

This paper has dealt with the problem of assessing the risk of APTs. Based on a state evolution model of an organization, the risk of the organization is measured by its maximum expected loss, and the risk assessment problem is modeled as a constrained optimization problem. Our theoretical study expounds the way that different factors affect an organization's risk. We speculate from experiments that the attack strategy obtained by applying the hill-climbing method to any instance of the proposed optimization problem leads to the maximum expected loss. Comparative experiments support our conjecture. The impact of two factors on the attacker's cost profit is determined through computer simulations.

There are many open problems toward this direction. This work builds on the premise that the defense posture is fixed. To enhance the security of an organization, the cyber defender may well flexibly adjust the defense posture over time. In this context, the optimal control theory provides an appropriate framework for developing cost-effective defense strategies \cite{Khouzani2012a, Khouzani2012b, ChenPY2014, ChenPY2015, YangLX2016, ZhangTR2017}. In situations where the attacker and defender are both strategic, it is feasible to assess the risk of APTs in the framework of game theory \cite{Alpcan2011, Khouzani2012c, LiangXN2013, HuPF2015}. In this work, the network of an organization is assumed to be fixed. In reality, this network may well vary over time \cite{Schwarzkopf2010, Valdano2015, Karyotis2015, Sanatkar2016, Cho2016}. So, it is of importance to assess the risk of APTs in this context. The identification of propagation resources in complex networks is a hotspot of research in the field of cyber security \cite{JiangJJ2015, YangF2016, Manitz2017, JiangJJ2017}. We suggest to utilize the state evolution model established in this work to identify the footholds of the attacker in a network. Also, it is rewarding to extend this work to the more realistic scenarios where queuing networks are involved \cite{Karyotis2008}. In recent years, cloud computing has been extended to the edge of organizational networks, forming fog computing \cite{YiSH2015, Ivan2015, Alrawais2017, Khan2017, Roman2018}. In this context, the assessment of the risk of APTs must be a huge challenge.

\section*{Acknowledgments}

The authors are grateful to the anonymous reviewers and the editor for their valuable comments and suggestions, which have greatly improved the quality of the paper. This work is supported by National Natural Science Foundation of China (Grant No. 61572006) and Sci-Tech Support Program of China (Grant No. 2015BAF05B03).


%\end{IEEEbiography}


%\newpage



%\vfill


\end{document}